%% file: thesis.tex
\documentclass
[
    twoside,                 
    openright,               
    cleardoublepage = empty, 
    fontsize = 12 pt,        
    american,                
    captions = tableheading, 
    numbers = noenddot,      
    footheight = 35 pt,      
]
{scrbook}

\input{core/preamble/document_settings} 

    \extraBorder{3 mm}

    \bindingCorrection{6 mm}

    \myTitle{Diameter Shortcut Sets in Temporal Graphs}{Diameter Shortcut Sets in Temporalen Graphen}

    \myName{Gerome Quantmeyer}

    \myProgram{IT Systems Engineering}

    \dateOfHandingIn{07. März 2025}

    \nameOfMySupervisor{Prof.\,Dr. Tobias Friedrich}

    \additionalExaminers{Dr. George Skretas\newline Michelle Döring}

    \mySubject{Exploring various approaches to diameter shortcut sets in temporal graphs.}

    \myKeywords{bachelor thesis | temporal graphs | diameter shortcut sets | graph theory | computer science}

\input{core/preamble/format}                        
\input{core/bibliography/bibliography_format}       
\input{packages_and_commands/standard_packages}     
\input{packages_and_commands/general_commands}      
\input{packages_and_commands/user-defined_commands} 

\begin{document}

    \frontmatter
    \include{core/title_page/title_page}

    \pagestyle{plain}

    \addchap{Abstract}
    \input{preface/abstract}

    \selectlanguage{ngerman}
    \addchap{Zusammenfassung}
    \input{preface/abstractGerman}
    \selectlanguage{american}

    \addchap{Acknowledgments}
    \input{preface/acknowledgments}

    \setuptoc{toc}{totoc}
    \tableofcontents

    \pagestyle{headings}
    \mainmatter

    \chapter{Introduction}
    \input{introduction/introduction}

    \chapter{Related Work}
    \input{introduction/related_work}
    \chapter{Preliminaries}
    \input{introduction/preliminaries}
    
    \chapter{Temporal Diameter Shortcut Sets}
    \input{chapters/tdss_definition}
    \chapter{Temporal Diameter Shortcut Sets on Temporal Paths}
    \input{chapters/tdss_paths}
    \chapter{Temporal Diameter Shortcut Sets using Static Expansions}
    \input{chapters/tdss_expansion}

    \makeatletter
        \def\toclevel@chapter{-1}
        \def\toclevel@section{0}
    \makeatother

    \chapter{Conclusions \& Outlook}
    \input{conclusions/conclusions}

    \pagestyle{plain}

    \renewcommand*{\bibfont}{\small}
    \printbibheading
    \addcontentsline{toc}{chapter}{Bibliography}
    \printbibliography[heading = none]

    \addchap{Declaration of Authorship}
    \input{core/declaration_of_authorship/declaration_of_authorship}

\end{document}

%% file: core/preamble/document_settings.tex
\newif\ifprintVersion   
\newif\ifprofessionalPrint 
\newif\iffancyTheorems  
\newif\ifboldNumberSets 
\newif\ifbachelorThesis 

\printVersionfalse
\professionalPrintfalse
\fancyTheoremstrue
\boldNumberSetstrue
\bachelorThesistrue

\newcommand*{\printTitle}{}
\newcommand*{\printGermanTitle}{}
\newcommand*{\myTitle}[2]{\renewcommand*{\printTitle}{#1}\renewcommand*{\printGermanTitle}{#2}}
\newcommand*{\printTitleBold}{\textbf{\printTitle}}

\newcommand*{\printAuthor}{}
\newcommand*{\myName}[1]{\renewcommand*{\printAuthor}{#1}}

\newcommand*{\printProgram}{}
\newcommand*{\myProgram}[1]{\renewcommand*{\printProgram}{#1}}

\newcommand*{\printDateReceived}{}
\newcommand*{\dateOfHandingIn}[1]{\renewcommand*{\printDateReceived}{#1}}

\newcommand*{\printSubject}{}
\newcommand*{\mySubject}[1]{\renewcommand*{\printSubject}{#1}}

\newcommand*{\printKeywords}{}
\newcommand*{\myKeywords}[1]{\renewcommand*{\printKeywords}{#1}}

\newcommand*{\printNameOfSupervisor}{}
\newcommand*{\nameOfMySupervisor}[1]{\renewcommand*{\printNameOfSupervisor}{#1}}

\newcommand*{\printAdditionalExaminers}{}
\newcommand*{\additionalExaminers}[1]{\renewcommand*{\printAdditionalExaminers}{#1}}

\newlength{\extraborderlength}
\newcommand*{\extraBorder}[1]{\setlength{\extraborderlength}{#1}}

\newlength{\mybindingcorrection}
\newcommand*{\bindingCorrection}[1]{\setlength{\mybindingcorrection}{#1}}

%% file: core/preamble/format.tex
\usepackage[utf8]{inputenc} 
\usepackage[T1]{fontenc}    
\usepackage
[
    ngerman,         
    main = american, 
]
{babel}                     

\input glyphtounicode
\usepackage{calc} 

\newlength{\myparindent}
\newlength{\myparskip}
\setfootnoterule{0 cm}

\deffootnote[1.2 em]{1.2 em}{0 em}{\makebox[1.4 em][l]{\textbf{\thefootnotemark}}}

\makeatletter%
    \@removefromreset{footnote}{chapter}%
\makeatother

\usepackage[dvipsnames]{xcolor} 

\definecolor{stroke1}{HTML}{2574A9} 

\colorlet{captionlabel}{black}
\colorlet{footerpagenr}{black}
\colorlet{footerchapter}{stroke1}
\colorlet{footerchaptername}{black}
\colorlet{footersection}{stroke1}
\colorlet{footersectionname}{black}
\colorlet{chapternumber}{stroke1}

\newlength{\mypaperwidth}
\newlength{\mypaperheight}
\newlength{\mybodywidth}
\newlength{\mybodyheight}
\newlength{\myoutermargin}
\ifprintVersion
    \ifprofessionalPrint
    \else
    \fi
\else
\fi

\newlength{\mytopmargin}
\ifprintVersion
    \ifprofessionalPrint
    \fi
\fi

\newlength{\myinnermargin}
\ifprintVersion
    \ifprofessionalPrint
    \fi
\fi

\newlength{\mybottommargin}
\ifprintVersion
    \ifprofessionalPrint
    \fi
\fi

\newcommand{\goldenratio}{1.618}

\newlength{\myheadsep} 
\ifprintVersion
    \ifprofessionalPrint
    \fi
\fi

\newlength{\myfootskip} 
\ifprintVersion
    \ifprofessionalPrint
    \fi
\fi

\newlength{\mymargininnersep} 
\newlength{\mymarginoutersep} 
\ifprintVersion
    \ifprofessionalPrint
    \fi
\fi

\newlength{\mymarginwidth} 
\newlength{\mymarginwidthwithinnersep} 
\usepackage
[
    \ifprintVersion
        \ifprofessionalPrint
            paperwidth = \mypaperwidth + 2\extraborderlength + \mybindingcorrection,
            paperheight = \mypaperheight + 2\extraborderlength,
        \else
            paperwidth = \mypaperwidth,
            paperheight = \mypaperheight,
        \fi
    \else
        paperwidth = \mypaperwidth,
        paperheight = \mypaperheight,
    \fi
    textwidth = \mybodywidth,
    textheight = \mybodyheight,
    outer = \myoutermargin,
    top = \mytopmargin,
    headsep = \myheadsep,
    footskip = \myfootskip,
    marginparsep = \mymargininnersep,
    marginparwidth = \mymarginwidth,
]
{geometry} 

\usepackage
[
]
{scrlayer-scrpage} 

\clearpairofpagestyles

\KOMAoptions
{%
    headwidth = \textwidth + \mymarginwidthwithinnersep,%
    footwidth = \myoutermargin : \textwidth,%
}

\automark[chapter]{chapter}
\automark*[section]{}

\lehead%
{%
    \begin{minipage}[b]{\mymarginwidth}%
        \small\raggedleft\normalfont\textsf{\textbf{\color{footerchapter}\chaptername\ \thechapter}}
    \end{minipage}
}
\cehead{\hspace*{\mymarginwidthwithinnersep}\parbox{\textwidth}{\raggedright\leftmark}}

\rohead%
{%
    \Ifstr{\rightmark}{\leftmark}%
    {%
        \begin{minipage}[b]{\mymarginwidth}%
            \small\raggedright\normalfont\textsf{\textbf{\color{footersection}Chapter\ \thechapter}}%
        \end{minipage}%
    }%
    {%
        \begin{minipage}[b]{\mymarginwidth}%
            \small\raggedright\normalfont\textsf{\textbf{\color{footersection}Section\ \thesection}}%
        \end{minipage}%
    }%
}
\cohead{\hspace*{-\mymarginwidthwithinnersep}\parbox{\textwidth}{\raggedleft\rightmark}}

\lefoot*%
{%
    \vspace*{1 ex}%
    {\color{stroke1}\rule{\myoutermargin - \mymargininnersep}{0.5 mm}}\\
    \begin{minipage}[b]{\myoutermargin - \mymargininnersep}%
        \raggedleft\normalfont\color{footerpagenr}\textbf{\thepage}%
    \end{minipage}%
}
\rofoot*%
{%
    {\color{stroke1}\rule{\myoutermargin - \mymargininnersep}{0.5 mm}}\\
    \begin{minipage}[b]{\myoutermargin - \mymargininnersep}%
        \raggedright\normalfont\color{footerpagenr}\textbf{\thepage}%
    \end{minipage}%
}

\usepackage{caption}
\usepackage{tikz} 

\newlength{\mytmpa}
\newlength{\mytmpb}

\renewcommand*{\partlineswithprefixformat}[3]%
{%
    #2
    \thispagestyle{empty}
    \setlength{\mytmpa}{0.618\mypaperwidth}%
    \setlength{\mytmpb}{0.382\mypaperheight}%
    \ifprintVersion
        \ifprofessionalPrint
            \setlength{\mytmpa}{0.618\mypaperwidth + \mybindingcorrection + \extraborderlength}%
            \setlength{\mytmpb}{0.382\mypaperheight + \extraborderlength}%
        \fi
    \fi
    \begin{tikzpicture}[overlay, remember picture]%
        \node [inner sep = 0, outer sep = 0, anchor = north] at (current page.north west)%
        {%
            \begin{tikzpicture}[overlay, remember picture]%
            \draw[color = stroke1, line width = 0.7 mm] (\mytmpa, 0) -- (\mytmpa, -\mytmpb);%
            \end{tikzpicture}%
        };%
        \node (align) [align = right, below = \mytmpb - 2 ex, inner sep = 0, outer sep = 0, anchor = north west] at (current page.north west)%
        {%
            \hspace{\mytmpa}\hspace{0.5 em}\partname\ \thepart\\[1 ex]
            \color{stroke1}#3%
        };%
    \end{tikzpicture}%
}
\RedeclareSectionCommand%
[%
    font = \normalfont\Huge\sffamily,
    prefixfont = \normalfont\Huge\sffamily,
]
{part}

\usepackage{etoolbox}

\newbool{chapterHasANumber}
\newbool{chapterHasAStar}
\renewcommand*{\chapterlinesformat}[3]%
{%
    \Ifnumbered{#1}{\setbool{chapterHasANumber}{true}}{\setbool{chapterHasANumber}{false}}%
    \Ifstr{#2}{}{\setbool{chapterHasAStar}{true}}{\setbool{chapterHasAStar}{false}}%
    \ifboolexpr{bool{chapterHasANumber} and not bool{chapterHasAStar}}%
    {%
        \begin{tikzpicture}[overlay, remember picture]%
            \node [right = \myinnermargin, below = \mytopmargin, inner sep = 0, outer sep = 0, anchor = north west] (numbernode) at (current page.north west)%
            {%
                \hspace{\myinnermargin}%
                \sffamily\fontsize{60}{60}\selectfont%
                \color{chapternumber}%
                \thechapter%
            };%
            \node [inner sep = 0, outer sep = 0, anchor = north west] at (numbernode.south west)%
            {%
                \begin{tikzpicture}[overlay, remember picture]%
                    \draw[color = stroke1, line width = 0.7 mm] (\myinnermargin, -1 ex) -- (\paperwidth, -1 ex);%
                \end{tikzpicture}%
            };%
            \node (align) [text width = \textwidth - 2 cm, align = right, right = \myinnermargin + \mybodywidth, inner sep = 0, outer sep = 0, anchor = east] at (numbernode.west)%
            {%
                #3%
            };%
        \end{tikzpicture}%
    }%
    {%
        \begin{tikzpicture}[overlay, remember picture]%
            \node [right = \myinnermargin, below = \mytopmargin, inner sep = 0, outer sep = 0, anchor = north west] (numbernode) at (current page.north west)%
            {%
                \hspace{\myinnermargin}%
                \sffamily\fontsize{60}{60}\selectfont%
                \color{white}%
                \thechapter%
            };%
            \node [inner sep = 0, outer sep = 0, anchor = north west] at (numbernode.south west)%
            {%
                \begin{tikzpicture}[overlay, remember picture]%
                    \draw[color = stroke1, line width = 0.7 mm] (\myinnermargin, -1 ex) -- (\paperwidth, -1 ex);%
                \end{tikzpicture}%
            };%
            \node (align) [align = left, right = \myinnermargin, inner sep = 0, outer sep = 0, anchor = south west] at (numbernode.south west)%
            {%
                #3%
            };%
        \end{tikzpicture}%
    }%
}
\RedeclareSectionCommand%
[%
    font = \color{stroke1}\normalfont\huge\sffamily,
    afterskip = 20 pt,
]
{chapter}

\BeforeStartingTOC[toc]{\pagestyle{plain}}
\AfterStartingTOC{\thispagestyle{plain}}

%% file: core/bibliography/bibliography_format.tex
\usepackage
[
    sortcites,              
    style = alphabetic,     
    defernumbers,           
    safeinputenc,           
    backref = true,         
    backrefstyle = three,   
    hyperref = true,        
    maxbibnames = 99,       
    maxcitenames = 2,       
]
{biblatex} 

\renewbibmacro{in:}%
{%
    \ifentrytype{article}{}{\printtext{\bibstring{in}\intitlepunct}}%
}

\renewbibmacro*{volume+number+eid}%
{%
    \printfield{volume}%
    \iffieldundef{number}{}{\addcolon}%
    \printfield{number}%
    \setunit*{\addcomma\space}%
    \printfield{eid}%
}

\DefineBibliographyStrings{english}%
{%
    backrefpage  = {\lowercase{s}ee page}, 
    backrefpages = {\lowercase{s}ee pages} 
}

\DeclareFieldFormat[article]{title}{\textbf{\color{stroke1}#1}}
\DeclareFieldFormat[inproceedings]{title}{\textbf{\color{stroke1}#1}}
\DeclareFieldFormat[thesis]{title}{\textbf{\color{stroke1}#1}}
\DeclareFieldFormat[book]{title}{\textbf{\color{stroke1}#1}}
\DeclareFieldFormat[unpublished]{title}{\textbf{\color{stroke1}#1}}
\DeclareFieldFormat[report]{title}{\textbf{\color{stroke1}#1}}
\DeclareFieldFormat[inbook]{chapter}{\textbf{\color{stroke1}#1}}
\DeclareFieldFormat[inbook]{title}{#1}
\DeclareFieldFormat{pages}{#1}

\newtoggle{authorend}
\togglefalse{authorend}

\DeclareBibliographyDriver{article}%
{%
  \usebibmacro{bibindex}%
  \usebibmacro{begentry}%
  \iftoggle{authorend}{}{\usebibmacro{author/translator+others}}%
  \setunit{\labelnamepunct}\newblock
  \usebibmacro{title}%
  \newunit
  \printlist{language}%
  \newunit\newblock
  \usebibmacro{byauthor}%
  \newunit\newblock
  \usebibmacro{bytranslator+others}%
  \newunit\newblock
  \printfield{version}%
  \newunit\newblock
  \usebibmacro{in:}%
  \usebibmacro{journal+issuetitle}%
  \newunit
  \usebibmacro{byeditor+others}%
  \newunit
  \usebibmacro{note+pages}%
  \newunit\newblock
  \iftoggle{bbx:isbn}
  {\printfield{issn}}
  {}%
  \newunit\newblock
  \usebibmacro{doi+eprint+url}%
  \newunit\newblock
  \usebibmacro{addendum+pubstate}%
  \setunit{\bibpagerefpunct}\newblock
  \usebibmacro{pageref}%
  \newunit\newblock
  \iftoggle{bbx:related}
  {\usebibmacro{related:init}%
    \usebibmacro{related}}
  {}%
  \usebibmacro{finentry}%
  \iftoggle{authorend}{\usebibmacro{author/translator+others}}{}%
}

\DeclareBibliographyDriver{inbook}%
{%
  \usebibmacro{bibindex}%
  \usebibmacro{begentry}%
  \iftoggle{authorend}{}{\usebibmacro{author/translator+others}}%
  \setunit{\labelnamepunct}\newblock
  \usebibmacro{chapter+pages}%
  \newunit
  \printlist{language}%
  \newunit\newblock
  \usebibmacro{byauthor}%
  \newunit\newblock
  \usebibmacro{in:}%
  \usebibmacro{bybookauthor}%
  \newunit\newblock
  \usebibmacro{maintitle+booktitle}%
  \newunit\newblock
  \usebibmacro{byeditor+others}%
  \newunit\newblock
  \printfield{edition}%
  \newunit
  \iffieldundef{maintitle}
  {\printfield{volume}%
    \printfield{part}}
  {}%
  \newunit
  \printfield{volumes}%
  \newunit\newblock
  \usebibmacro{series+number}%
  \newunit\newblock
  \printfield{note}%
  \newunit\newblock
  \usebibmacro{publisher+location+date}%
  \newunit\newblock
  \newunit\newblock
  \iftoggle{bbx:isbn}
  {\printfield{isbn}}
  {}%
  \newunit\newblock
  \usebibmacro{doi+eprint+url}%
  \newunit\newblock
  \usebibmacro{addendum+pubstate}%
  \setunit{\bibpagerefpunct}\newblock
  \usebibmacro{pageref}%
  \newunit\newblock
  \iftoggle{bbx:related}
  {\usebibmacro{related:init}%
    \usebibmacro{related}}
  {}%
  \usebibmacro{finentry}%
  \iftoggle{authorend}{\usebibmacro{author/translator+others}}{}%
}

\DeclareBibliographyDriver{inproceedings}%
{%
  \usebibmacro{bibindex}%
  \usebibmacro{begentry}%
  \iftoggle{authorend}{}{\usebibmacro{author/translator+others}}%
  \setunit{\labelnamepunct}\newblock
  \usebibmacro{title}%
  \newunit
  \printlist{language}%
  \newunit\newblock
  \usebibmacro{byauthor}%
  \newunit\newblock
  \usebibmacro{in:}%
  \usebibmacro{maintitle+booktitle}%
  \newunit\newblock
  \usebibmacro{event+venue+date}%
  \newunit\newblock
  \usebibmacro{byeditor+others}%
  \newunit\newblock
  \iffieldundef{maintitle}
  {\printfield{volume}%
    \printfield{part}}
  {}%
  \newunit
  \printfield{volumes}%
  \newunit\newblock
  \usebibmacro{series+number}%
  \newunit\newblock
  \printfield{note}%
  \newunit\newblock
  \printlist{organization}%
  \newunit
  \usebibmacro{publisher+location+date}%
  \newunit\newblock
  \usebibmacro{chapter+pages}%
  \newunit\newblock
  \iftoggle{bbx:isbn}
  {\printfield{isbn}}
  {}%
  \newunit\newblock
  \usebibmacro{doi+eprint+url}%
  \newunit\newblock
  \usebibmacro{addendum+pubstate}%
  \setunit{\bibpagerefpunct}\newblock
  \usebibmacro{pageref}%
  \newunit\newblock
  \iftoggle{bbx:related}
  {\usebibmacro{related:init}%
    \usebibmacro{related}}
  {}%
  \usebibmacro{finentry}%
  \iftoggle{authorend}{\usebibmacro{author/translator+others}}{}%
}

\DeclareBibliographyDriver{thesis}%
{%
  \usebibmacro{bibindex}%
  \usebibmacro{begentry}%
  \iftoggle{authorend}{}{\usebibmacro{author}}%
  \setunit{\labelnamepunct}\newblock
  \usebibmacro{title}%
  \newunit
  \printlist{language}%
  \newunit\newblock
  \usebibmacro{byauthor}%
  \newunit\newblock
  \printfield{note}%
  \newunit\newblock
  \printfield{type}%
  \newunit
  \usebibmacro{institution+location+date}%
  \newunit\newblock
  \usebibmacro{chapter+pages}%
  \newunit
  \printfield{pagetotal}%
  \newunit\newblock
  \iftoggle{bbx:isbn}
  {\printfield{isbn}}
  {}%
  \newunit\newblock
  \usebibmacro{doi+eprint+url}%
  \newunit\newblock
  \usebibmacro{addendum+pubstate}%
  \setunit{\bibpagerefpunct}\newblock
  \usebibmacro{pageref}%
  \newunit\newblock
  \iftoggle{bbx:related}
  {\usebibmacro{related:init}%
    \usebibmacro{related}}
  {}%
  \usebibmacro{finentry}%
  \iftoggle{authorend}{\usebibmacro{author}}{}%
}

\providebool{bbx:subentry}
\newbibmacro*{citenum}%
{
  \printtext[bibhyperref]{
    \printfield{prefixnumber}%
    \printfield{labelnumber}%
    \ifbool{bbx:subentry}
    {\printfield{entrysetcount}}
    {}}%
}

\DeclareCiteCommand{\conline}[\mkbibbrackets]
{\usebibmacro{prenote}}
{\usebibmacro{citeindex}%
  \usebibmacro{citenum}}
{\multicitedelim}
{\usebibmacro{postnote}}

%% file: packages_and_commands/standard_packages.tex
\usepackage
[
    babel = true, 
]
{microtype}           
\usepackage{csquotes} 

\usepackage{amsmath}
\usepackage{amssymb}
\usepackage{amsthm}
\usepackage{thmtools}
\usepackage{mathtools}
\usepackage{thm-restate}
\usepackage{dsfont}        
\usepackage{braceMnSymbol} 

\usepackage
[
    ttscale = 0.85, 
]
{libertine} 
\usepackage
[
    libertine,    
    slantedGreek, 
    vvarbb,       
    libaltvw,     
]
{newtxmath} 
\usepackage{url} 
\usepackage{bm}  

\usepackage{graphicx} 
\usepackage
[
    subrefformat = simple, 
    labelformat = simple,  
]
{subcaption}         
\usepackage{wrapfig} 

\columnsep = \mymargininnersep

\usepackage{array}     
\usepackage{booktabs}  
\usepackage{longtable} 
\usepackage{pdflscape} 

\usepackage{enumitem} 

\usepackage
[
    ruled,         
    vlined,        
    linesnumbered, 
]
{algorithm2e} 

\usepackage{xspace}   
\usepackage
[
    shortcuts, 
]
{extdash}             
\usepackage{setspace} 

\usepackage{xparse}    
\usepackage{footnote}  
\usepackage{afterpage} 
\usepackage
[
    textsize = scriptsize, 
]
{todonotes}            

\input{packages_and_commands/additional_packages}

\usepackage
[
    bookmarks = true,                 
    bookmarksopen = false,            
    bookmarksnumbered = true,         
    pdfstartpage = 1,                 
    pdftitle = {{\printTitle}},       
    pdfauthor = {{\printAuthor}},     
    pdfsubject = {{\printSubject}},   
    pdfkeywords = {{\printKeywords}}, 
    breaklinks = true,                
    \ifprintVersion
        hidelinks,                    
    \else
    colorlinks = true,            
    allcolors = stroke1,          
    \fi
]
{hyperref} 

\usepackage
[
    noabbrev,   
    nameinlink, 
]
{cleveref} 

%% file: packages_and_commands/additional_packages.tex
\usepackage{float}

%% file: packages_and_commands/general_commands.tex
\newcommand*{\colloquialDegreeName}{Master}
\newcommand*{\colloquialDegreeNameLowercase}{master}

\newcommand*{\degreeAbbreviation}{M.}

\ifbachelorThesis
    \renewcommand*{\colloquialDegreeName}{Bachelor}
    \renewcommand*{\colloquialDegreeNameLowercase}{bachelor}
    \renewcommand*{\degreeAbbreviation}{B.}
\fi

\makeatletter
    \def\IfEmptyTF#1%
    {%
        \if\relax\detokenize{#1}\relax%
            \expandafter\@firstoftwo%
        \else%
            \expandafter\@secondoftwo%
        \fi%
    }
\makeatother

\NewDocumentCommand{\mathOrText}{m}
{%
    \ensuremath{#1}\xspace%
}

\let\originalleft\left
\let\originalright\right
\renewcommand{\left}{\mathopen{}\mathclose\bgroup\originalleft}
\renewcommand{\right}{\aftergroup\egroup\originalright}

\makeatletter
    \DeclareRobustCommand{\bfseries}%
    {%
        \not@math@alphabet\bfseries\mathbf%
        \fontseries\bfdefault\selectfont%
        \boldmath%
    }
\makeatother

\xspaceaddexceptions{]\}}

\allowdisplaybreaks

\crefname{ineq}{inequality}{inequalities}
\creflabelformat{ineq}{#2{\upshape(#1)}#3} 

\crefname{term}{term}{terms}
\creflabelformat{term}{#2{\upshape(#1)}#3}
\crefname{section}{section}{sections}
\creflabelformat{section}{#2{\upshape #1}#3}
\crefname{chapter}{chapter}{chapters}
\creflabelformat{chapter}{#2{\upshape #1}#3}
\crefname{subsection}{subsection}{subsections}
\creflabelformat{subsection}{#2{\upshape #1}#3}

\let\oldfootnote\footnote

\newlength{\spaceBeforeFootnote} 
\newlength{\spaceAfterFootnote}  

\RenewDocumentCommand{\footnote}{o o o m}%
{%
    \IfNoValueTF{#1}%
    {%
        \oldfootnote{#4}%
    }%
    {%
        \setlength{\spaceBeforeFootnote}{\IfEmptyTF{#1}{0}{#1} em}%
        \IfNoValueTF{#2}%
        {%
            \hspace*{\spaceBeforeFootnote}\oldfootnote{#4}%
        }%
        {%
            \setlength{\spaceAfterFootnote}{\IfEmptyTF{#2}{0}{#2} em}%
            \hspace*{\spaceBeforeFootnote}\IfNoValueTF{#3}{\oldfootnote{#4}}{\oldfootnote[#3]{#4}}\hspace*{\spaceAfterFootnote}%
        }%
    }%
}

\makesavenoteenv{figure}
\makesavenoteenv{table}
\makesavenoteenv{tabular}

\iffancyTheorems
    \declaretheoremstyle
    [
        spaceabove = \topsep,
        spacebelow = \topsep,
        headfont = \bfseries,
        headformat = \textcolor{stroke1}{$\blacktriangleright$} \NAME~\NUMBER \NOTE,
        notefont = \bfseries,
        notebraces = {(}{)},
        bodyfont = \normalfont,
        postheadspace = 0.5 em,
        qed = \textcolor{stroke1}{\bfseries$\blacktriangleleft$},
    ]
    {myTheoremStyle}
    
    \declaretheorem
    [
        style = myTheoremStyle,
        name = Conjecture,
        numberwithin = chapter,
    ]
    {conjecture}
    \declaretheorem
    [
        style = myTheoremStyle,
        name = Proposition,
        sharenumber = conjecture,
    ]
    {proposition}
    \declaretheorem
    [
        style = myTheoremStyle,
        name = Claim,
        sharenumber = conjecture,
    ]
    {claim}
    \declaretheorem
    [
        style = myTheoremStyle,
        name = Lemma,
        sharenumber = conjecture,
    ]
    {lemma}
    \declaretheorem
    [
        style = myTheoremStyle,
        name = Corollary,
        sharenumber = conjecture,
    ]
    {corollary}
    \declaretheorem
    [
        style = myTheoremStyle,
        name = Theorem,
        sharenumber = conjecture,
    ]
    {theorem}
    \declaretheorem
    [
        style = myTheoremStyle,
        name = Definition,
        sharenumber = conjecture,
    ]
    {definition}
    \declaretheorem
    [
        style = myTheoremStyle,
        name = Example,
        sharenumber = conjecture,
    ]
    {example}
    \declaretheorem
    [
        style = myTheoremStyle,
        name = Remark,
        sharenumber = conjecture,
    ]
    {remark}
\else
    \theoremstyle{plain}

\fi

\NewDocumentCommand{\functionTemplate}{m m m m o}%
{%
    \IfNoValueTF{#5}%
    {%
        \mathOrText{#1\left#2{#4}\right#3}%
    }%
    {%
        \mathOrText{#1#5#2{#4}#5#3}%
    }%
}

\newcommand*{\leftBracketType}{(}
\newcommand*{\rightBracketType}{)}

\NewDocumentCommand{\createFunction}{m m o o}%
{%
    \renewcommand*{\leftBracketType}{\IfNoValueTF{#3}{(}{#3}}%
    \renewcommand*{\rightBracketType}{\IfNoValueTF{#4}{)}{#4}}%
    \NewDocumentCommand{#1}{o o}%
    {%
        \IfNoValueTF{##1}%
        {%
            \mathOrText{#2}%
        }%
        {%
            \functionTemplate{#2}{\leftBracketType}{\rightBracketType}{##1}[##2]%
        }%
    }%
}

\DeclareDocumentCommand{\probabilisticFunctionTemplate}{m m O{} o}
{%
    \functionTemplate{#1}%
    {\lbrack}%
    {\rbrack}%
    {#2\IfEmptyTF{#3}{}{\ \IfNoValueTF{#4}{\left}{#4}\vert\ \vphantom{#2}#3\IfNoValueTF{#4}{\right.}{}}}%
    [#4]%
}

\ifboldNumberSets
    \newcommand*{\N}{\mathOrText{\mathbb{N}}}

    \newcommand*{\indicatorFunctionSymbol}{\mathbf{1}}
\else
    \newcommand*{\N}{\mathOrText{\mathds{N}}}

    \newcommand*{\indicatorFunctionSymbol}{\mathds{1}}
\fi

\RenewDocumentCommand{\Pr}{m O{} o}%
{%
    \probabilisticFunctionTemplate{\mathrm{Pr}}{#1}[#2][#3]%
}

\NewDocumentCommand{\Ex}{m O{} o}%
{%
    \probabilisticFunctionTemplate{\mathrm{E}}{#1}[#2][#3]%
}

\NewDocumentCommand{\Var}{m O{} o}%
{%
    \probabilisticFunctionTemplate{\mathrm{Var}}{#1}[#2][#3]%
}

\DeclareDocumentCommand{\bigO}{m o}%
{%
    \functionTemplate{\mathrm{O}}{(}{)}{#1}[#2]%
}

\DeclareDocumentCommand{\smallO}{m o}%
{%
    \functionTemplate{\mathrm{o}}{(}{)}{#1}[#2]%
}

\DeclareDocumentCommand{\bigTheta}{m o}%
{%
    \functionTemplate{\upTheta}{(}{)}{#1}[#2]%
}

\DeclareDocumentCommand{\bigOmega}{m o}%
{%
    \functionTemplate{\upOmega}{(}{)}{#1}[#2]%
}

\DeclareDocumentCommand{\smallOmega}{m o}%
{%
    \functionTemplate{\upomega}{(}{)}{#1}[#2]%
}

\DeclareDocumentCommand{\bigOlog}{m o}
{
    \functionTemplate{\Tilde{\mathrm{O}}}{(}{)}{#1}[#2]
}

\DeclareDocumentCommand{\bigOmegalog}{m o}
{
    \functionTemplate{\Tilde{\upOmega}}{(}{)}{#1}[#2]
}

\DeclareDocumentCommand{\eulerE}{o}%
{%
    \mathOrText{\mathrm{e}\IfNoValueTF{#1}{}{^{#1}}}%
}

\DeclareDocumentCommand{\poly}{m o}%
{%
    \functionTemplate{\mathrm{poly}}{(}{)}{#1}[#2]%
}

\createFunction{\id}{\mathrm{id}}

\NewDocumentCommand{\ind}{m o o}%
{%
    \IfNoValueTF{#2}%
    {%
        \mathOrText{\indicatorFunctionSymbol_{#1}}%
    }%
    {%
        \functionTemplate{\indicatorFunctionSymbol_{#1}}{(}{)}{#2}[#3]%
    }%
}

\DeclareDocumentCommand{\dom}{m o}%
{%
    \functionTemplate{\mathrm{dom}}{(}{)}{#1}[#2]%
}

\DeclareDocumentCommand{\rng}{m o}%
{%
    \functionTemplate{\mathrm{rng}}{(}{)}{#1}[#2]%
}

\DeclareDocumentCommand{\d}{o}%
{%
    \mathrm{d}\IfNoValueTF{#1}{}{^{#1}}%
}

\DeclareDocumentCommand{\set}{m m o}%
{
    \mathOrText{\IfNoValueTF{#3}{\left}{#3}\{#1\ \IfNoValueTF{#3}{\left}{#3}\vert\
    \vphantom{#1}#2\IfNoValueTF{#3}{\right.}{}\IfNoValueTF{#3}{\right}{#3}\}}
}

%% file: packages_and_commands/user-defined_commands.tex
\newcommand*{\tempG}{\mathOrText{\mathcal{G}}}
\newcommand*{\tempP}{\mathOrText{\mathcal{P}}}
\newcommand*{\tempE}{\mathOrText{\mathcal{E}}}

\newcommand{\staticTrans}{\mathOrText{G_S}}

\newcommand{\resultTempG}{\mathOrText{\tempG^*}}
\newcommand*{\TDSS}{\mathOrText{\mathcal{S}}}

%% file: core/title_page/title_page.tex

\ifprintVersion
    \ifprofessionalPrint
        \newgeometry
        {
            textwidth = 134 mm,
            textheight = 220 mm,
            top = 38 mm + \extraborderlength,
            inner = 38 mm + \mybindingcorrection + \extraborderlength,
        }
    \else
        \newgeometry
        {
            textwidth = 134 mm,
            textheight = 220 mm,
            top = 38 mm,
            inner = 38 mm + \mybindingcorrection,
        }
    \fi
\else
    \newgeometry
    {
        textwidth = 134 mm,
        textheight = 220 mm,
        top = 38 mm,
        inner = 38 mm,
    }
\fi

\begin{titlepage}
    \sffamily
    \begin{center}
        \includegraphics[height = 3.2 cm]{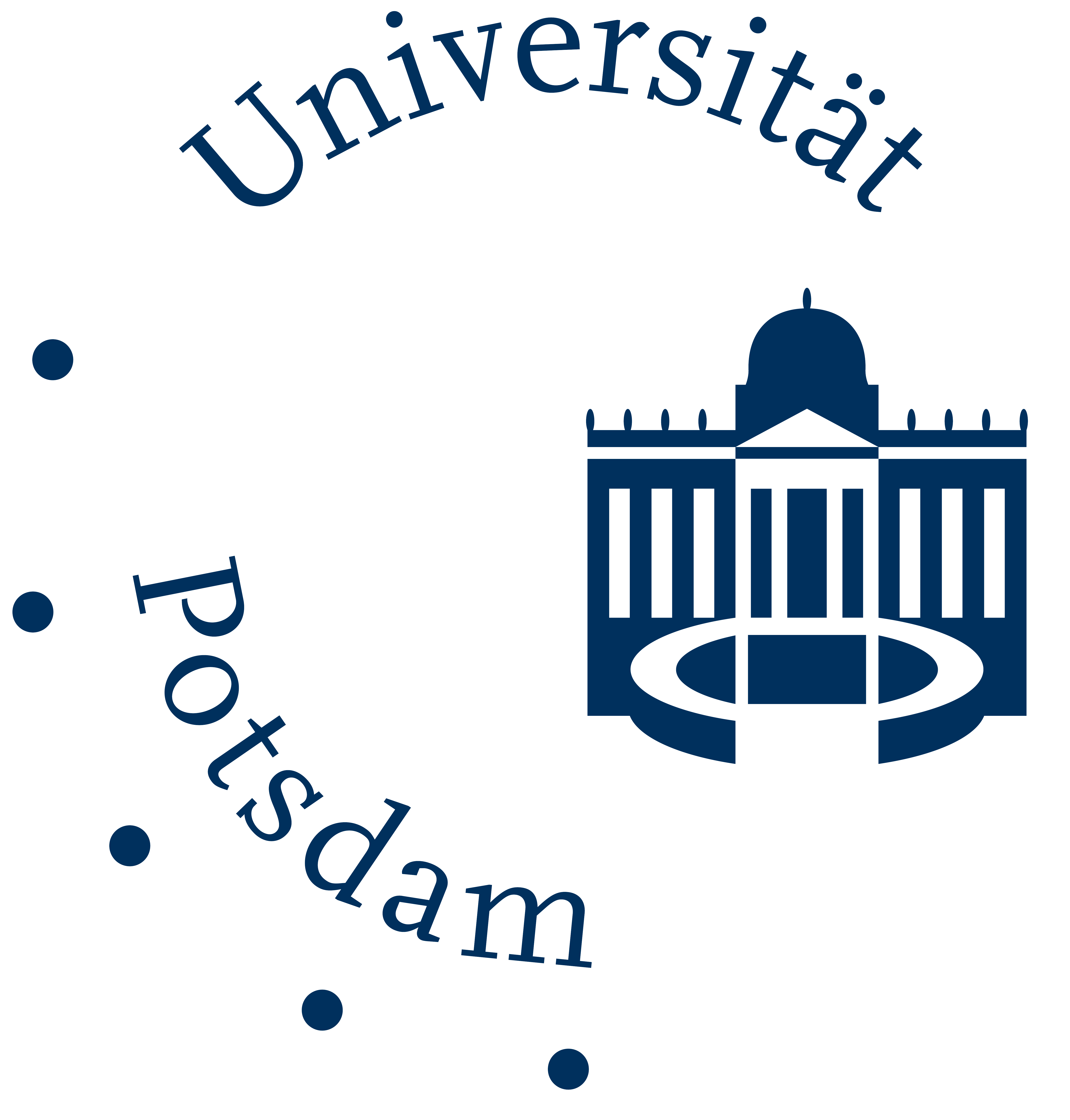} \hfill \includegraphics[height = 3 cm]{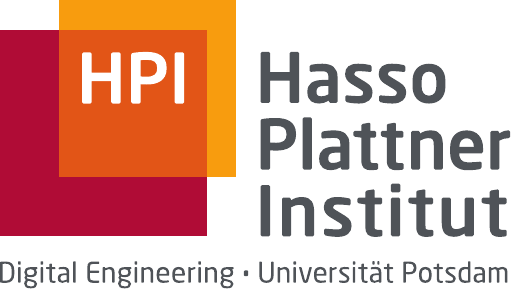}\\
        \vfil
        {\LARGE
            \rule[1 ex]{\textwidth}{1.5 pt}
            \onehalfspacing\printTitleBold\\[1 ex]
            {\vspace*{-1 ex}\Large \printGermanTitle}\\
            \rule[-1 ex]{\textwidth}{1.5 pt}
        }
        \vfil
        {\Large\textbf{\printAuthor}}
        \vfil
        {\large Universitäts\colloquialDegreeNameLowercase arbeit\\[0.25 ex]
        zur Erlangung des akademischen Grades}\\[0.25 ex]
        \bigskip
        {\Large \colloquialDegreeName{} of Science}\\[0.5 ex]
        {\large\emph{(\degreeAbbreviation\,Sc.)}}\\
        \bigskip
        {\large im Studiengang\\[0.25 ex]
        \printProgram}
        \vfil
        {\large eingereicht am \printDateReceived{} am\\[0.25 ex]
        Fachgebiet Algorithm Engineering der\\[0.25 ex]
        Digital-Engineering-Fakultät\\[0.25 ex]
        der Universität Potsdam}
    \end{center}
    
    \vfil
    \begin{table}[h]
        \centering
        \large
        \sffamily 
        {\def\arraystretch{1.2}
            \begin{tabular}{>{\bfseries}p{3.8 cm}p{5.3 cm}}
                Gutachter               & \printNameOfSupervisor\\
                Betreuer                & \printAdditionalExaminers
            \end{tabular}
        }
    \end{table}
\end{titlepage}

\restoregeometry

%% file: preface/abstract.tex
Shortcut sets are a vital instrument for reducing the diameter of a static graph and, consequently, its shortest path complexity, which is relevant in numerous subfields of graph theory. We explore the notion of shortcut sets in temporal graphs, which incorporate a discrete time model into the graph, rendering each edge accessible exclusively at specific points in time. This not only alters the underlying assumptions of regular graphs but also substantially increases the complexity of path problems and reachability. In turn, a temporal graph is often a much more realistic and accurate representation of a real-world network. In this thesis we provide a definition for a shortcut set in a temporal graph and explore differences to classic shortcut sets. Utilizing this definition, we show that temporal and regular shortcut sets yield the same results on temporal paths, enabling the application of existing construction algorithms for static shortcut sets on paths. The primary contribution of this thesis is a translation approach for general temporal graphs that utilizes the static expansion of a temporal graph, allowing the conversion of static shortcut sets into temporal shortcut sets, yielding similar results.

%% file: preface/abstractGerman.tex
Shortcut Sets sind ein wichtiges Instrument zur Verringerung des Durchmessers eines statischen Graphen und folglich der Komplexität von kürzesten Wegen, was in zahlreichen Teilgebieten der Graphentheorie von Bedeutung ist. Wir erforschen das Konzept von Shortcut Sets in temporalen Graphen, die ein diskretes Zeitmodell in den Graphen einbeziehen, so dass jede Kante nur zu bestimmten Zeitpunkten nutzbar ist. Dies verändert nicht nur die zugrundeliegenden Annahmen für reguläre Graphen, sondern erhöht auch die Komplexität von Pfadproblemen und der Erreichbarkeit erheblich. Im Gegenzug ist ein temporaler Graph oft eine viel realistischere und genauere Darstellung eines realen Netzwerks. In dieser Arbeit geben wir eine Definition für ein Shortcut Set in einem temporalen Graphen und untersuchen die Unterschiede zu klassischen Shortcut Sets. Anhand dieser Definition zeigen wir, dass temporale und reguläre Shortcut Sets die gleichen Ergebnisse auf temporalen Pfaden liefern, was die Anwendung bestehender Konstruktionsalgorithmen für statische Shortcut Sets auf Pfaden ermöglicht. Der primäre Beitrag dieser Arbeit ist ein Übersetzungsansatz für allgemeine temporale Graphen, der die statische Erweiterung eines temporalen Graphen nutzt und die Umwandlung von statischen Shortcut Sets in temporale Shortcut Sets ermöglicht, was zu vergleichbaren Ergebnissen führt.

%% file: preface/acknowledgments.tex
First of all, I would like to thank my supervisors, Dr. George Skretas and Michelle Döring, for their invaluable support while working on this thesis. They have contributed important insights and support to this work, pushing me in the right direction to arrive at the results that I now present in this thesis.

Michelle, in particular, contributed to this thesis with her continuous and rigorous feedback throughout the writing process, encouraging me to continually improve the final product. She was also always available to support me and answer even the simplest questions that came up.

I particularly appreciated George's detached approach, allowing me to develop my own ideas and then stepping in with useful insights, questions, and critiques, while also being available for more in-depth discussion when needed. This provided a great balance of creative freedom and insightful support.

I would also like to thank the other people who took the time to review my work and provide valuable feedback. These include Vincent Quantmeyer, Ben Bals, and Hendrik Higl.

%% file: introduction/introduction.tex
\label{chapter:introduction}
The complexity of shortest path problems plays a crucial role in various computational models, often being heavily influenced by the diameter of the underlying graph. Additionally, the runtime of distributed or parallel algorithms often depends on the diameter. One approach to reducing the diameter of a graph is through the use of a diameter shortcut set, which expands a given graph with additional edges, called shortcuts. The goal is to reduce the diameter of a graph as much as possible, while adding as few shortcuts as possible.

A shortcut set $S$ is a set of edges selected from the transitive closure of a given graph $G=(V,E)$ such that the modified graph $(V, E\cup S)$ has a reduced diameter. The transitive closure $TC(G)=(V, E^*)$ of a static graph is defined as the graph where $(u,v)\in E^*$ if and only if $v$ is reachable from $u$. The concept of shortcut sets was first introduced by \textcite{thorupShortcuttingDigraphs1992} and has been extensively studied in the context of static graphs over the past decades.

In this thesis, we explore the idea of diameter shortcut sets in the domain of temporal graphs\footnote{Throughout this work, we refer to standard (i.e., non-temporal) graphs as static graphs to distinguish them from temporal graphs.}. To the best of our knowledge, the concept of diameter shortcut sets has not yet been studied in the context of temporal graphs.

Temporal graphs extend static graphs by incorporating a discrete time model, representing scenarios where edges become available only at specific time steps. Temporal graphs can be used to model real-world phenomena, such as transportation networks, social networks, and communication systems. Formally, a temporal graph consists of a static graph $G = (V,E)$, where each edge is labeled with a finite set of numbers. These labels indicate the discrete time points at which an edge is active, which could correspond to real-world time units such as seconds, minutes, or days. Importantly, any static edge can be available at multiple time steps, with each availability referred to as a temporal edge. $G$ is also known as the footprint of the temporal graph.

In this thesis, we propose a definition for a diameter shortcut set in temporal graphs and present some general observations implied by the nature of temporal graphs. Our primary contribution is a translation approach that maps a static shortcut set onto a temporal graph via a static expansion. This transformation allows us to leverage existing results on shortcut sets in static graphs to derive corresponding results for temporal graphs.
\section*{Temporal Graphs}\label{section:introTemporal}
We give a brief introduction into the field of temporal graphs and aim to give an overview of basic properties and observations for temporal graphs.
\begin{figure}
    \centering
    \includegraphics[width=0.7\textwidth]{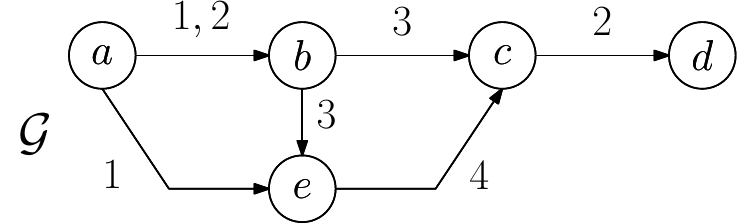}
    \caption{Example of a temporal graph}
    \label{fig:exampleTemporalGraph}
\end{figure}

Temporal graphs can be modeled in multiple different ways. All of these approaches are equivalent and are used interchangably. In some cases a specific model may be easier to work with than the others.
\begin{enumerate}
    \item The simplest solution is to use a labeling function $t: E\to 2^\N$ that assigns each edge a (possibly empty) set of time labels. We can now refer to our temporal graph based on $t$ as $t(G)$. This approach makes analyzing the labeling itself the easiest. For example using $t$ we can easily define the minimum label in a graph $t_{min}$ or the maximum label $t_{max}$. We could also refer to all labels of a single edge $e$ as $t(e)$. This information allows us to calculate the lifetime $\alpha(t) = t_{max}-t_{min}+1t$. Similar to the lifetime we can define the lifecycle $c(t) = [t_{min},t_{max}]$ as the range of labels.
    \item Another approach is to redefine the set of edges $E$ to $\tempE$. Let $\tempG=(V,\tempE)$ be a temporal graph. Then $e\in \tempE$ is a triple $(u,v,t)$ or sometimes even a quadruple $(u,v,t,\lambda)$. Here $u,v\in V$, $t\in\N$ is the starting time of the edge\footnote{Note that labels could also be chosen from other sets of numbers, though $\N$ is most common.} and $\lambda$ is the traversal time. If we do not want to regard traversal time we can either omit $\lambda$ or set $\lambda = 1$ for all edges $e\in \tempE$. Sometimes a temporal edge is also denoted as $(e,t)\in\tempE$, where $e$ is just a regular static edge.
    \newpage\item We can also interpret a set of static graphs as a temporal graph, where each static graph represents a time label $t$ and the edges active at that time. We call each of these graphs a snapshot $\tempG_t$ of the temporal graph $\tempG$.
\end{enumerate}

Finally, we can represent a temporal graph as an equivalent static graph, called a static expansion. Every node $v\in V$ expands into a set of nodes $V_v$ (one per time label). We connect a node $v_t\in V_v$ to the node representing the next time label $v_{t+1}\in V_v$. Now we can add an edge between two vertices of different sets $V_v$ and $V_u$, which then represents a temporal edge at the time of the starting vertex.

Static expansions have been used multiple times throughout literature, for example by \textcite{michailIntroductionTemporalGraphs2015} or \textcite{wuPathProblemsTemporal2014}. With this approach it is possible to use some static algorithms and solutions to compute solutions for the temporal graph. For example, \textcite{wuPathProblemsTemporal2014} used a static expansion in combination with algorithms like Dijkstra's Algorithm to compute shortest temporal paths. Depending on the problem, some variation on the exact definition of the static expansion may be necessary. In this thesis, we will use a variation of the static expansion to compute temporal diameter shortcut sets later.

\begin{figure}
    \centering
    \includegraphics[width=0.8\textwidth]{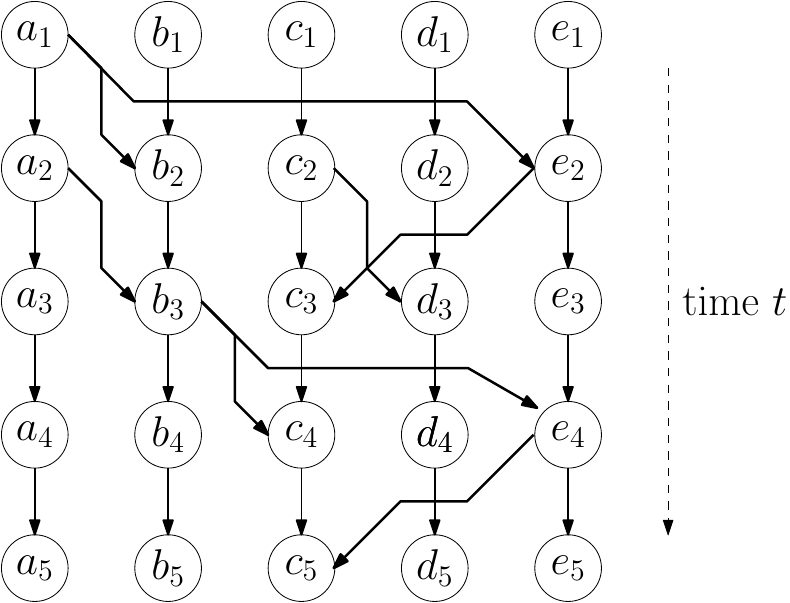}
    \caption{Static Expansion of the temporal graph $\tempG$ (see \Cref{fig:exampleTemporalGraph})}
    \label{fig:exampleClassicStaticExpansion}
\end{figure}

Note that with any of these models, temporal graphs can be both directed and undirected. An undirected temporal graph can easily be represented as a directed graph with every undirected edge replaced by two directed edges with the same time label, but not all problems (e.g., Eulerian trails) are equivalent under this reduction.
\subsection*{Temporal Reachability}\label{subsection:introTemporalReachability}
Temporal journeys introduce some more complexities compared to static journeys. In addition to a static path on the footprint of the temporal graph, the time labels need to increase along a journey or path. We differentiate between strict and non-strict temporal journeys. Strict temporal journeys require strictly increasing time labels, whereas non-strict temporal journeys allow increasing or equal time labels (see \Cref{fig:exampleStrictNonStrict}).
\begin{figure}
    \centering
    \includegraphics[width=0.8\textwidth]{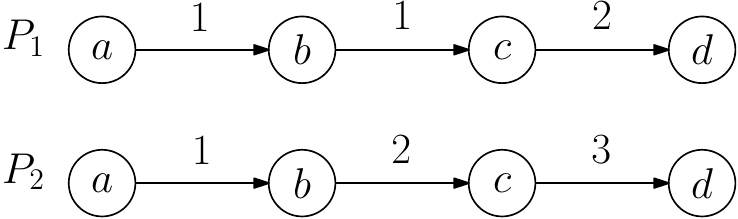}
    \caption{Example of a non-strict ($P_1$) and strict ($P_2$) Temporal Journey}
    \label{fig:exampleStrictNonStrict}
\end{figure}
Of course every strict temporal journey is also a non-strict temporal journey. This additional restriction makes a major difference for the reachability in a temporal graph.
\begin{restatable}[Transitivity of Reachability]{corollary}{transitivityOfReachability}
    \label{crl:transitivityOfReachability}
    Given a static graph $G=(V,E)$ and $a,b,c\in V$. If $a$ reaches $b$ and $b$ reaches $c$, then $a$ also reaches $c$. 
\end{restatable}
\Cref{crl:transitivityOfReachability} holds in a static graph and is the basis of many algorithms for path problems like Dijkstra's algorithm. Due to the added time constraint, reachability in temporal graphs is not transitive (see \Cref{fig:transitivityCounterExample}). The lack of transitivity of temporal reachability will become relevant later in this thesis.
\begin{figure}
    \centering
    \includegraphics[width=0.5\textwidth]{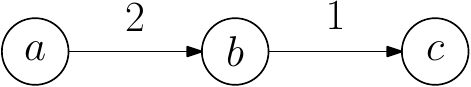}
    \caption{Transitivity of Reachability}
    \label{fig:transitivityCounterExample}
\end{figure}

When examining the problem of shortest journeys, there are multiple metrics to consider. Let us look at the real-world example of a train network. We can optimize the time the journey takes as a whole, the time spent traveling (and not waiting), the number of different trains taken, the arrival time or the departure time. For example, \textcite{wuPathProblemsTemporal2014} introduce earliest arrival paths, latest departure paths, fastest paths and shortest paths. These correspond to optimizing arrival time, departure time, lifetime of the path and hop distance respectively. Depending on the metric, the method of finding an optimal path can change, although there are some metrics that relate. For example, the earliest arrival time given a fixed departure time is also the fastest time.

%% file: introduction/related_work.tex
\label{chapter:relatedWork}
We present related work categorized into two different sections. First, we go over the most recent and most relevant results for static shortcut sets, as our main contributions in this thesis are largely based on these results. This includes both computational results and approaches to the construction of static shortcut sets. Afterwards we explore constructions and concepts similar to shortcut sets in the context of temporal graphs.

\section{Static Shortcut Sets}
\label{section:staticRelatedWork}
A static shortcut set $S$ is a set of edges taken from the transitive closure of a graph $G=(V,E)$, so that the graph $(V,E\cup S)$ has reduced diameter. A shortcut set that achieves diameter $d$ is sometimes referred to as a $d$-shortcut set. The notion of a shortcut set was first introduced by \textcite{thorupShortcuttingDigraphs1992}. The term \emph{transitive closure spanner} or \emph{$d$-TC-spanner} was later introduced by \textcite{bhattacharyyaTransitiveClosureSpanners2008} and refers to the same concept. We will therefore treat them both as static shortcut sets here.

There are two settings that are addressed continuously throughout the literature for static shortcut sets. Given a directed unweighted graph $G=(V,E)$ with $|V| = n$ nodes and $|E|=m$, by how much can a shortcut set (1) of size $\bigO{n}$ or (2) of size $\bigO{m}$ reduce the diameter of $G$? Both settings preserve the sparsity of the graph, within a constant factor. 

The most recent results for upper bounds on the two scenarios are based on a randomized construction by \textcite{ullmanHighprobabilityParallelTransitive1990}, which was later improved by \textcite{koganNewDiameterReducingShortcuts2022}. \textcite{bermanFindingSparserDirected2010} also provides an approximation approach, contributing to the state of the art results. 

For lower bounds on static shortcut sets, we again address the same two scenarios as above. For $\bigO{m}$ shortcuts, the best results were published by \textcite{luBetterLowerBounds2022}. They improved on the constructions of \textcite{hesseDirectedGraphsRequiring2003} and \textcite{huangLowerBoundsSparse2018}. \textcite{bodwinFolkloreSamplingOptimal2024} provide the best lower bound for $\bigO{n}$ shortcuts added.
\newpage The two scenarios yield the following results:
\begin{enumerate}[noitemsep]
    \item A possible diameter of $\bigO{n^{1/3}}$ when adding $\bigOlog{n}$ shortcuts to the graph with a lower bound of $\bigOmegalog{n^{1/4}}$ and
    \item a possible diameter of $\bigO{(n^2/m)^{1/3}}$ when adding $\bigOlog{m}$ shortcuts with a lower bound of $\bigOmega{n^{1/5}}$.
\end{enumerate}

A separate problem to tackle is the construction of shortcut sets given some desired diameter $d$. As mentioned above, \textcite{ullmanHighprobabilityParallelTransitive1990} introduced a randomized approach to construct a shortcut set, which works as follows. Given a graph with $n$ vertices and a desired diameter $d$, we choose a set $S$ of $\bigO{n\log n/d}$ vertices and add the transitive closure of $S$ to our graph. This achieves a diameter of $\bigO{d}$ \cite{williamsSimplerHigherLower2023}.

\textcite{bermanFindingSparserDirected2010} also introduced an approximation approach resulting in two construction algorithms for small and large values of $d$. For small $d = \bigO{\log n/\log\log n}$ a transformation from $2$-shortcut sets to $d$-shortcut sets is used in combination with an $\bigO{\log n}$-approximation for $2$-shortcut sets by \textcite{elkinClientServer2SpannerProblem2001}. This resulted in an approximation ratio of $\bigO{n^{1-1/\lceil d/2\rceil}\log n}$. For large $d=\bigOmega{\log n/\log\log n}$ \textcite{bermanFindingSparserDirected2010} presented \Cref{fig:staticSpannerAlg}\footnote{\Cref{fig:staticSpannerAlg} uses the transitive reduction of a graph $G$, which is just the smallest possible subgraph on the same vertex set, that has the same reachabilities as $G$.}. This algorithm resulted in an approximation ratio of $\bigO{n/d^2}$. 

\begin{algorithm}
    \KwIn{directed graph $G=(V,E)$, desired diameter $d$}
    Let $H$ be a transitive reduction of $G$\;
    \While{$H$ contains vertices $u,v$ such that $dist_H(u,v)>d$}{
        Let $(v_1=u,v_2,...,v_t=v)$ be the shortest path from $u$ to $v$ in $H$\;
        Add a shortcut edge ($v_{\lfloor d/4\rfloor}, v_{t-\lfloor d/4\rfloor})$ to $H$\;
    }
    Output $H$.
    \caption{$\bigO{n/d^2}$-Approximation Algorithm for $d$-shortcut sets}
    \label{fig:staticSpannerAlg}
\end{algorithm}

\section{Related Concepts in Temporal Graphs}\label{section:temporalRelatedWork}
In this section we explore different concepts in temporal graphs that are either similar to shortcut sets or correlate to our results in some other way. We begin by looking at constructions that manipulate a temporal graph by adding or removing labels or edges to achieve a certain metric, like connectivity or distances. This is of course very similar to the concept of a diameter shortcut set.

A similar field to shortcut sets in temporal graphs is \emph{temporal graph realization}. This field has been studied, among others, in \cite{erlebachParameterizedAlgorithmsMultiLabel2024}, \cite{mertziosRealizingTemporalTransportation2024} and \cite{klobasRealizingTemporalGraphs2024}. Given a static graph $G$ or a distance matrix, one tries to create a temporal graph that realizes a given metric as well as possible. For example, \textcite{klobasRealizingTemporalGraphs2024} explored temporal graph realization based on fastest path distances. More specifically given an $n\times n$ matrix $D$ and a period\footnote{A periodic temporal graph assigns time labels in a periodic range, which is then repeated.} $\Delta\in\N$, we try to construct a $\Delta$-periodic temporal graph with $n$ vertices such that the duration of a fastest path between two vertices $v_i$ and $v_j$ is equal to $D_{i,j}$. 

Another similar topic is label manipulation. \textcite{deligkasMinimizingReachabilityTimes2021} for example explored the delay of connections in temporal graphs to improve the overall travel time between vertices. Given a temporal graph $\tempG$ and a set of source vertices, they ask how we can minimize the maximum time needed to reach every vertex in $\tempG$ by delaying some number of temporal edges. When allowed to use an unbounded number of delays, this can be computed in polynomial time given only one source vertex. In any other case the problem becomes $NP$-hard.

\textcite{klobasComplexityComputingOptimum2022} also explored the \emph{Minimum Labeling Problem}, which asks for the smallest number of time labels $|\lambda|$ for a graph $G$ such that the resulting temporal graph is temporally connected. A temporal graph is temporally connected if all vertices are reachable through some temporal path. The basic problem of minimum labeling can be optimally solved in polynomial time, although some variations of the problem, also explored by Klobas at al. are $NP$-hard.

All of the above mentioned work can be categorized as \emph{temporal design problems}. We either, given a temporal graph, try to modify its temporal edge set to achieve some metric or, given some dataset or static base graph, we try to create a temporal graph from it to again optimize some graph metric.

Lastly, \textcite{whitbeckTemporalReachabilityGraphs2012} explored the concept of temporal reachability graphs. A $(\tau,\delta)$-\emph{reachability graph} based on a temporal graph $\tempG=(V,\tempE)$, is a temporal graph where for $u,v\in V$ there exists a temporal edge in the reachability graph at time $t$ if there exists a temporal journey between $u$ and $v$ starting after $t$ with a traversal time of $\tau$ and a maximum delay of $\delta$. \textcite{whitbeckTemporalReachabilityGraphs2012} developed a theoretical framework around temporal reachability graphs and introduced algorithms for their computation.

A reachability graph in the static context is equivalent to a static transitive closure. We also explore possible definitions for temporal transitive closures in this thesis to use as a basis of a temporal diameter shortcut set.

%% file: introduction/preliminaries.tex
\label{chapter:preliminaries}
We will now formally introduce terms used throughout the rest of the thesis. We begin by defining the relevant terms for shortcut sets in static graphs.
\begin{restatable}[Static Transitive Closure]{definition}{staticTransitiveClosure}
    \label{def:staticTransitiveClosure}
    Given a static graph $G=(V,E)$ the \emph{transitive closure} of $G$, denoted as $TC(G)$, is the graph $(V,E^*)$ with $(u,v)\in E^*$ \emph{if and only if} $u$ is reachable from $v$ in $G$.
\end{restatable}
Using the static transitive closure we can now define the static shortcut set as follows:
\begin{restatable}[Static Shortcut Set]{definition}{staticShortcutSet}
    \label{def:staticShortcutSet}
    Given a static graph $G=(V,E)$ and its transitive closure $TC(G) = (V,E^*)$ a \emph{static shortcut set} $S\subseteq E^*$, is a set of edges, where the graph $(V,E\cup S)$ has reduced diameter. A shortcut set achieving a diameter $d$, is referred to as a $d$-shortcut set.
\end{restatable}
We call an edge $s\in S$ a \emph{shortcut} and may refer to adding a shortcut or shortcut set to a given graph as \emph{shortcutting}.

Additionally we introduce relevant terms for temporal graphs, beginning with a temporal graph itself.
\begin{restatable}[Temporal Graph]{definition}{temporalGraph}
    \label{def:temporalGraph}
    A \emph{temporal graph} $\tempG=(V,\tempE)$ is a graph, where $V$ is the set of vertices in $\tempG$ and $\tempE$ is the set of \emph{temporal edges} $e =(u,v,t)\in\tempE$. A temporal edge $e$ consists of the starting and ending vertices $u,v\in V$ as well as a \emph{starting time} $t\in\N$.
\end{restatable}
Note that we could also introduce a \emph{traversal time} $\lambda\in\N$ as part of a temporal edge. We assume a constant traversal time $\lambda =1$ for all edges and therefore omit $\lambda$ entirely.

Given $t_{min},t_{max}\in\N$ as the minimum and maximum label for any edge $e\in\tempE$, we define the \emph{lifetime} of a temporal graph $\tempG$ as $T(\tempG) = t_{max}-t_{min}+1$. Similarly, the \emph{lifecycle} refers to the range of time labels $[t_{min},t_{max}]$ in $\tempG$. We may refer to the underlying static graph $G=(V,E)$, where $E=\{(u,v)\mid (u,v,t)\in\tempE\}$, as the \emph{footprint} of $\tempG$. We may refer to a temporal edge $(u,v,t)$ as $(e,t)$ with $e = (u,v)$ being a static edge on the footprint of $\tempG$.

\newpage We define a temporal path, which is based on a simple static path, but adds some more temporal restrictions:
\begin{restatable}[Temporal Path]{definition}{temporalPath}
    \label{def:temporalPath}
    Given a temporal graph $\tempG=(V,\tempE)$, a \emph{temporal path} is a sequence of vertices $\tempP = \langle v_1,v_2,...,v_k,v_{k+1}\rangle$, where $(v_i,v_{i+1},t_i)\in\tempE$ and $t_i<t_{i+1}$ for all $1\leq i<k$.
\end{restatable}
For the purposes of this thesis we will work with \emph{strict temporal paths} (as defined in \Cref{def:temporalPath}), though we may address the non-strict case at some points. In a \emph{non-strict temporal path} edge labels can be equal or increase in topological order. Of course this implies that, any strict temporal path is also a non-strict temporal path. Note that a static path on the footprint of $\tempG$ may contain multiple temporal paths, as any static edge may contain multiple temporal edges.

In addition to the standard diameter of a graph or temporal graph we also define a \emph{base diameter} based on temporal graph $\tempG$ as follows:
\begin{restatable}[Base Diameter]{definition}{baseDiameter}
    \label{def:baseDiameter}
    Let $\tempG = (V,\tempE)$ be a temporal graph. For any temporal graph $\resultTempG$ with the same vertex set $V$, the \emph{base diameter} $d_\tempG(\resultTempG)$ refers to the longest distance from any vertex $u\in V$ to any vertex $v\in V$ in $\resultTempG$, where $v$ was also reachable from $u$ in $\tempG$.
\end{restatable}
Of course the base diameter $d_\tempG(\tempG)$ is just the standard diameter of $\tempG$.

%% file: chapters/tdss_definition.tex
\label{chapter:tdssDefinition}
In this chapter we will discuss the construction of a Temporal Diameter Shortcut Set. At the end of the chapter we provide a definition for the concept, that we also use throughout the rest of the thesis.

Recall that a static shortcut set $S$ is a set of edges taken from the transitive closure of a graph $G=(V,E)$, so that $(V,E\cup S)$ has reduced diameter (see \Cref{def:staticShortcutSet}). The definition of the static shortcut set is based on the static transitive closure (\Cref{def:staticTransitiveClosure}). Among other things, this implies that no new reachabilities can be introduced to $G$ through the  addition of shortcuts from $S$. 

As we will discuss, due to the nature of temporal graphs, our definition will differ in some ways from that of static shortcut sets. More specifically our final definition will not be based on a temporal transitive closure and will not constrain the reachabilities added by the shortcuts. Nevertheless, we will explore the application of a transitive closure and similar concepts in temporal graphs to come to this conclusion. We will show that through the usage of any of these concepts new reachabilities will be added to the temporal graph.

\section{Temporal Transitive Closure}\label{section:temporalTransitiveClosure}
In this section we explore different possible definitions of a temporal transitive closure. For this we look at similar known concepts in static graphs.

In the static context a transitive closure preserves the reachability of its underlying graph. We share some concepts analog to the static transitive closure and conclude that all of them increase reachability in temporal graphs making them unfeasible as the basis of a temporal shortcut set.

Firstly, any temporal transitive closure will only use edges from the static transitive closure of the footprint. This is obviously implied by the restriction of reachability of a transitive closure. The question then becomes, which labels should belong in a temporal transitive closure. This has to be restricted somehow, so that the transitive closure is still finite.

\begin{enumerate}
    \item The simplest solution would be to allow any time label within the lifecycle of the temporal graph. This would allow shortcut sets to completely change different time metrics like earliest arrival time and latest departure time, which we would like to avoid.
    \item We could also choose the labels according to inherent time metrics of the graph like the earliest arrival time. More specifically, given a temporal graph $\tempG = (V,\tempE)$ the \emph{EAT Transitive Closure} would be a graph $\tempG^{TC} = (V,\tempE^{TC}\cup\tempE)$, where $(u,v,t)\in\tempE^{TC}$ \emph{if and only if} the earliest arrival time from $u$ to $v$ is $t$.
    \item Analog to 2. we could define a \emph{LDT Transitive Closure} $\tempG^{TC} = (V,\tempE^{TC}\cup\tempE)$ based on the latest departure time, where $(u,v,t)\in\tempE^{TC}$ \emph{if and only if} the latest departure time from $u$ to $v$ is $t$.
    \item Another option is to use a temporal equivalent to the definition of the static transitive closure. Recall that in the static transitive closure $TC(G)=(V,E^*)$ based on a static graph $G=(V,E)$  there exists an edge $(u,v)\in E^*$ if and only if $v$ is reachable from $u$ in $G$, i.e., there exists a path from $u$ to $v$. Of course the same definition could be used for a temporal graph instead requiring temporal paths from $u$ to $v$.
\end{enumerate}

\begin{figure}[h!]
    \centering
    \includegraphics[width=0.8\textwidth]{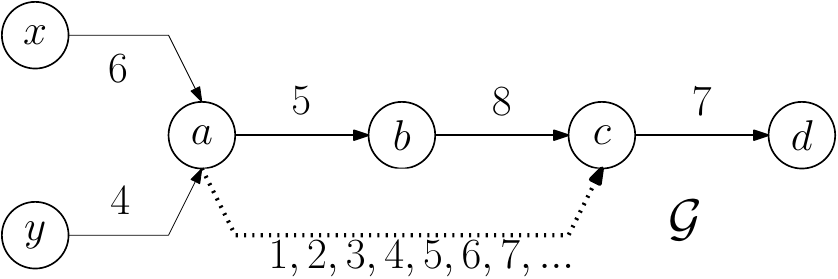}
    \caption{Counter Example for the Rechability Constraint}
    \label{fig:reachabilityConstraintCounterExample}    
\end{figure}
For all the approaches presented above, there exist simple examples, where reachability is increased. We provide one such example graph $\tempG = (V,\tempE)$ in \Cref{fig:reachabilityConstraintCounterExample}. We focus on a possible temporal shortcut $(a,c,t)$. If $t< 7$ then $a$ can reach $d$ and if $t\geq 7$ then $x$ can reach $c$, which are both new reachabilities. Given that $\tempG$ is a very simple example graph, we deduce that such a shortcut exists in many cases.
\begin{restatable}[]{corollary}{reachabilityIncreaseTTC}
    \label{crl:reachabilityIncreaseTTC}
    Given a temporal graph $\tempG=(V,\tempE)$ there may exist a temporal edge $(u,v,t)$ with $u$ and $v$ being reachable on the footprint of $\tempG$ that increases the reachabilities within $\tempG$.
\end{restatable}

This of course includes the edges that would be added through the approaches presented above. It also implies that there exist temporal graphs that can not be shortcut to any given diameter $d$, i.e., $d=1$, without adding new reachabilities.

Recall that reachability in temporal graphs is not transitive, i.e., a vertex $a$ reaching $b$ and $b$ reaching $c$ does not imply that $a$ can reach $c$ as well. Intuitively, this explains why many shortcuts add new reachabilities to a temporal graph and why the concept of a transitive closure in temporal graphs or an equivalent concept cannot work.

Approaches that try to add labels to existing edges instead of adding new edges run into a similar problem. Due to the lack of transitivity of reachability adding new labels to existing static connections can easily add new reachabilities to the temporal graph. For example one could easily circumvent a temporal break (\Cref{def:temporalBreak}) by adding a smaller time label $t$ to a static edge, as demonstrated in \Cref{fig:addLabelCounterExample}.
\begin{figure}[h!]
    \centering
    \includegraphics[width=0.8\textwidth]{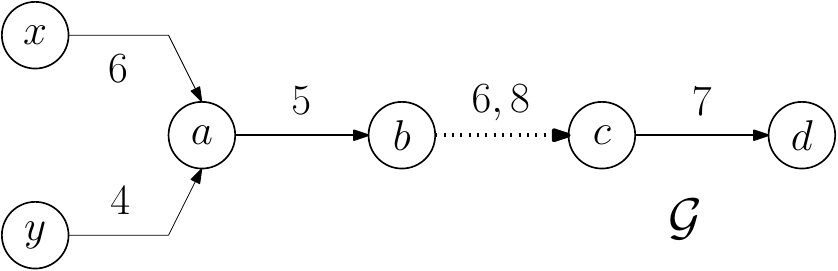}
    \caption{Counter Example for Reachability restriction, when adding labels to existing edges. The addition of $(b,c,6)$ adds new reachability.}
    \label{fig:addLabelCounterExample}
\end{figure}

We conclude that the restriction of reachability is not feasible in temporal graphs, when simply adding additional edges.
\section{Definition of a Temporal Diameter Shortcut Set}\label{section:temporalShortcutSet}
We explored different possible definitions of a temporal transitive closure and concluded that none of them imply a reachability constraint. We therefore do not base our definition on a temporal transitive closure and instead do not restrict a temporal shortcut by reachability. Of course for any temporal graph $\tempG$ there may exist temporal shortcut sets that add new reachabilities (see \Cref{fig:exampleTDSSnewReachabilities}) and ones that do not (see \Cref{fig:exampleTDSSnoRechabilites}).

Additionally, we do not regard the new reachabilities that are created through the temporal shortcut set and instead use the base diameter $d_\tempG(\tempG')$ as defined in \Cref{def:baseDiameter}, where $\tempG'$ is the shortcutted temporal graph. More specifically, if two vertices $u$ and $v$ are reachable after adding a temporal shortcut set $\TDSS$ and were not before, the distance between $u$ and $v$ may be larger than the base diameter $d_\tempG(\tempG')$ we achieve with $\TDSS$.
\begin{restatable}[Temporal Diameter Shortcut Set (TDSS)]{definition}{temporalShortcutSet}
    \label{def:temporalShortcutSet}
    Let $\tempG = (V,\tempE)$ be a temporal graph. A \emph{Temporal Diameter Shortcut Set} or just \emph{Temporal Shortcut Set} $\TDSS$ is a set of temporal edges $(u,v,t)$, where $u,v\in V$ and the base diameter $d_\tempG(\resultTempG=(V,\tempE\cup\TDSS))$ is reduced.\\
    We may refer to a Temporal Shortcut Set that achieves a base diameter $d_\tempG(\resultTempG)$ of at most $d$ in a temporal graph as a $d$-TDSS.
\end{restatable}

Analog to static shortcut sets, we call an edge $s\in\TDSS$ a \emph{shortcut} and may refer to adding a shortcut or shortcut set to a given temporal graph as \emph{shortcutting}. Note that for simplicity, if two newly reachable vertices $u$ and $v$ have distance larger than $d$, but $d_\tempG(\resultTempG)\leq d$, we still refer to $\TDSS$ as a $d$-TDSS.
\begin{figure}
    \centering
    \begin{subfigure}[b]{0.49\textwidth}
        \centering
        \includegraphics[width=\textwidth]{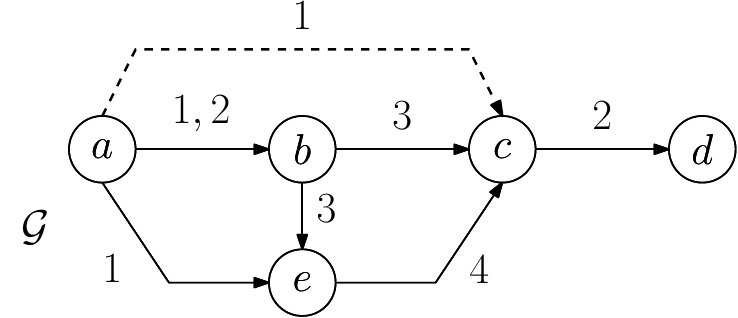}
        \caption{$2$-TDSS adding new reachabilities to $\tempG$}
        \label{fig:exampleTDSSnewReachabilities}
    \end{subfigure} 
    \hfil
    \begin{subfigure}[b]{0.49\textwidth}
        \centering
        \includegraphics[width=\textwidth]{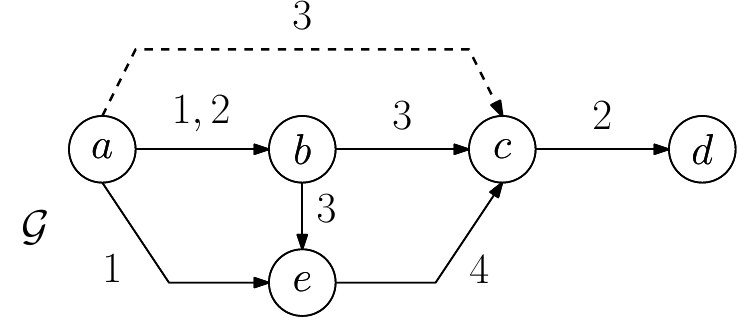}
        \caption{$2$-TDSS adding no reachabilities to $\tempG$}
        \label{fig:exampleTDSSnoRechabilites}
    \end{subfigure}
    \caption{Examples of TDSS on a temporal graph $\tempG$}
    \label{fig:examplesTDSS}
\end{figure}

The lack of constraint on the reachability added by a TDSS $\TDSS$ has some implications that may need to be considered.

Firstly, without the reachability constraint, it is possible to add cycles to a graph that previously did not have any. In a temporal graph with strict temporal journeys there cannot exist any temporal cycles. When examining non-strict temporal graphs this problem still exists, even if only in rare cases. 

Secondly and most importantly, we only regard the distances of vertices that were also previously reachable in $\tempG$. As described in \Cref{section:temporalTransitiveClosure}, many possible shortcuts add new reachabilities to the temporal graph. Considering the distances of these new reachabilties as well would seriously complicate the construction a temporal diameter shortcut set $\TDSS$. 

%% file: chapters/tdss_paths.tex
\label{chapter:tdssPaths}
Given our definition of a temporal diameter shortcut set (TDSS) from the last chapter, we start by looking at very simple graphs. More specifically we start with temporal graphs $\tempG = (V,\tempE)$ where the footprint of $\tempG$ is a directed path.

In this chapter we show that for any temporal graph $\tempG$ with a directed path as its footprint we can shortcut the temporal paths in $\tempG$ independently without increasing the size of the TDSS as a whole.

Additionally, we will prove that a TDSS on a directed temporal path achieves the same results as a static shortcut set on the footprint of that temporal path. This implies that to shortcut a directed temporal path we can fully utilize the constructions of static shortcut sets and the computational results regarding their relative size. Combining these two findings, we show that any temporal graph $\tempG$ with a directed path as its footprint, can be shortcut optimally using static shortcut sets.

Note that for the purposes of this chapter we will address simple temporal graphs, i.e, temporal graphs $\tempG=(V,\tempE)$, where for every static edge $(u,v)$ on the footprint only one temporal edge $(u,v,t)\in\tempE$ exists.

\section{Path Independence}\label{section:pathIndependence}
We start by showing that separate temporal paths can be shortcut independently. For that let $\tempG=(P,E)$ be a temporal graph with a directed path $P$ as its footprint. We first introduce the term \emph{temporal break}, which is a vertex $v_x$ on the footprint of a path, that separates the vertices before and after into two independent temporal paths. This is the case, because the time label before $v_x$ is larger than the time label after $v_x$, which constrains the reachability between the two parts, thus breaking them apart. Intuitively a temporal break is a vertex that breaks the transitivity of reachability.
\begin{restatable}[Temporal Break]{definition}{temporalBreak}
    \label{def:temporalBreak}
    Let $\tempG = (P,\tempE)$ be a temporal graph with a static path $P = \langle v_1,v_2,...,v_k,v_{k+1}\rangle$ as its footprint. Given $(v_{x-1},v_x,t_i),(v_x,v_{x+1},t_j)\in\tempE$, we call a vertex $v_x$ a \emph{temporal break} if $t_j\leq t_i$.
\end{restatable}
\newpage A temporal break implies two separate temporal paths on $P$ starting at $v_1$ and $v_x$ respectively (see \Cref{fig:footprintPathExamples}). By definition of a temporal path (\Cref{def:temporalPath}) the two temporal paths are not reachable from one another through $P$.

For temporal graphs $\tempG$ with directed paths as their footprint we can differentiate between two different types:
\begin{enumerate}
    \item[(a)] The static path is also a temporal path, i.e., the labels (strictly) increase in topological order (see \Cref{fig:noTemporalBreak}) or
    \item[(b)] there exists a temporal break at a vertex $v_x$, as defined in \Cref{def:temporalBreak} (see \Cref{fig:temporalBreak}).\footnote{There could of course also exist multiple temporal breaks at different vertices, separating the graph into more than two temporal paths.}
\end{enumerate}
\begin{figure}
    \centering
    \begin{subfigure}[b]{0.8\textwidth}
        \centering
        \includegraphics[width=\textwidth]{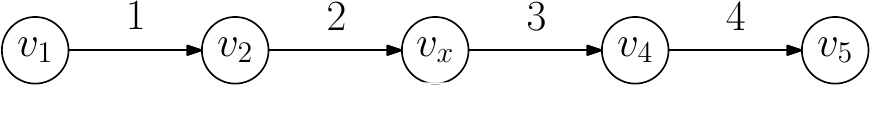}
        \caption{$\tempG$ without a temporal break}
        \label{fig:noTemporalBreak}
    \end{subfigure}
    \hfil
    \begin{subfigure}[b]{0.8\textwidth}
        \centering
        \includegraphics[width=\textwidth]{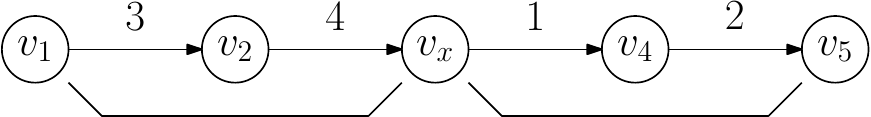}
        \caption{$\tempG$ with a temporal break at $v_x$}
        \label{fig:temporalBreak}
    \end{subfigure}
    \caption{Examples for (a) and (b) for $\tempG$ with a directed path as its footprint}
    \label{fig:footprintPathExamples}
\end{figure}

Inherently, the latter case implies that no vertices before $v_x\in V$ on the path can reach the vertices after $v_x$. Since $\tempG$ is directed, vertices after $v_x$ also cannot reach vertices before $v_x$. We deduce that we can add two disjoint temporal shortcut sets to the two subpaths implied by the temporal break at $v_x$ for a given desired base diameter $d$, which are not suboptimal in size.

Intuitively, adding a shortcut between the two implied subpaths can not reduce any distance on either of the subpaths. This implies that at least two shortcuts crossing from one subpath to the other are needed to achieve any distance reduction. Instead we could just add a single direct shortcut, which does not cross between the two subpaths.
\begin{restatable}[]{lemma}{pathIndependence}
    \label{lm:pathIndependence}
    Let $\tempG=(P,\tempE)$ be a temporal graph, where the footprint of $\tempG$ is a directed path $P = \{p_1,p_2,...,p_n\}$ with a temporal break in $\tempG$ at a vertex $p_x$. Let $S$ be a minimal temporal shortcut set on $\{p_1,p_2,...,p_x\}$ and $T$ be a minimal temporal shortcut set on $\{p_x,p_{x+1},...,p_n\}$. Then there exists no temporal shortcut set $U$ with $|U|<|S|+|T|$ that achieves the same base diameter as $S$ and $T$ in $\tempG$.
\end{restatable}
\begin{figure}[h!]
    \centering
    \includegraphics[width=\textwidth]{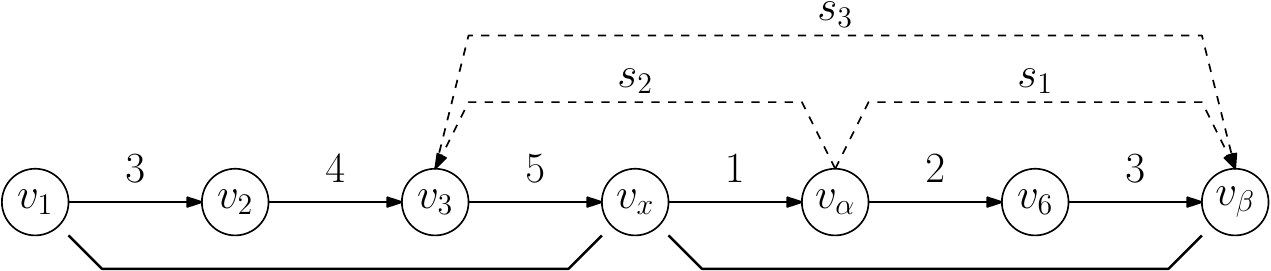}
    \caption{A temporal graph $\tempG$, where $s_1$ achieves a greater distance reduction than $s_2$ and $s_3$ combined between $v_\alpha$ and $v_\beta$.}
    \label{fig:pathIndependenceExample}
\end{figure}
\begin{proof}[Proof of \Cref{lm:pathIndependence}]
    Let $\tempG=(P,\tempE)$ be a temporal graph with a directed path $P = \{p_1,p_2,...,p_n\}$ as its footprint and a temporal break at $p_x$. Let $S$ be a minimal temporal shortcut set on $\{p_1,p_2,...,p_x\}$ and $T$ be a minimal temporal shortcut set on $\{p_x,p_{x+1},...,p_n\}$. Let $U$ be a temporal shortcut set on $\tempG$ and let $S$, $T$ and $U$ all achieve the same base diameter $d$ on their respective (sub)graphs.
    
    Firstly, if there exist no shortcuts $s_u=(p_i,p_j)\in U$ where either $i<x<j$ or $j<x<i$, we can partition $U$ into two disjoint shortcut sets on either subpath implied by the temporal break at $p_x$. Therefore $|U|<|S|+|T|$ can not hold, because $S$ and $T$ are minimal.

    We now assume that $|U|<|S|+|T|$. Let $s_u=(p_i,p_j)\in U$ now be a crossing shortcut as defined above. $s_u$  cannot reduce the distance between any two previously reachable vertices in $P$ on its own, since it does not connect two previously reachable vertices\footnote{Conversely $s_u$ even creates new rechabilities. Although we only regard the base diameter $d_\tempG((P,\tempE\cup U))$, avoiding this without loss of optimality would be desirable.}. Because $S$ and $T$ are minimal and $|U|<|S|+|T|$, the shortcuts crossing from one half of $P$ to the other need to decrease some distance between two previously reachable vertices. Given any set of such edges connecting two vertices $p_\alpha$ and $p_\beta$ with $\alpha,\beta>x$ w.l.o.g., adding a single shortcut $(p_\alpha,p_\beta)\in T$ achieves the same with fewer shortcuts added (see \Cref{fig:pathIndependenceExample}). Any distance reduction for vertices not in the same subpath are not necessary, as they are not reachable. Therefore, if no reachability exists and $S$ and $T$ are minimal, $|U|\geq|S|+|T|$ is implied.
\end{proof}

By \Cref{lm:pathIndependence} we can shortcut temporal paths implied by temporal breaks separately without increasing the size of the added temporal shortcut set. It may even be possible to partition more complex temporal graphs into temporal paths and shortcut them independently. 

For simple temporal graphs a temporal break always implies two temporal paths, where no vertex on one of the paths can reach a vertex on the other. Note that for non-simple temporal graphs, i.e., temporal graphs with more than one temporal edge per underlying static edge, this is not always the case. The lack of reachability between the two temporal paths is essential, which is of course also possible in non-simple temporal graphs.

\section{Static Shortcut Sets for Temporal Paths}\label{section:shortcutSetsOnTemporalPaths}
In this section we show that temporal shortcut sets on temporal paths actually achieve the same results as static shortcut sets. More specifically, given a static shortcut set $S$ on the footprint of a temporal path $\tempP$, we can easily add fitting time labels that achieve the same reachabilities and therefore the same diameter. Combined with the path independence explored in the previous section, this allows us to shortcut temporal graphs $\tempG$ with a directed path as their footprint using only static construction algorithms and also apply computational results for static shortcut sets to this subset of temporal graphs.

First of all, we can observe that a temporal path $\tempP$ and its underlying static path have the same reachabilities. If we now add a static $d$-shortcut set $S$ to the footprint of $\tempP$, we can achieve distance $d$ for all reachable vertices in $\tempP$ as well. All we need to do is add a set of time labels to the shortcuts $s\in S$ that do not restrict reachability.
\begin{restatable}[]{lemma}{addingReachableLabels}
    \label{lm:addingReachableLabels}
    Let $\tempP = \{p_1,p_2,...,p_n\}$ be a temporal path and $S$ be a static $d$-shortcut set on the footprint of $\tempP$. Then there exists a label $t_i\in\N$ for every shortcut $s_i=(p_i,p_j)\in S$, respectively, so that the temporal shortcut set $\TDSS = \{(s_i,t_i)|s_i\in S\}$ achieves diameter $d$.
\end{restatable}

Note that through the usage of a static shortcut set, we not only achieve a desired base diameter of $d$, but actually a standard diameter of $d$. 
\begin{proof}[Proof of \Cref{lm:addingReachableLabels}]
    Let $S$ be the static $d$-shortcut set for the footprint of $\tempP$. We want to construct a temporal shortcut set $\TDSS$ also achieving diameter $d$.
    
    Let $s=(p_i,p_j)\in S$. We search for a time label $t\in\N$ that does not restrict the usage of $s$ by adding a temporal to either $p_i$ or $p_j$. Since a static shortcut set preserves the reachabilities of the graph, every shortcut $s$ must travel in topological order of $\tempP$, i.e., $i<j$. Since $\tempP$ is a temporal path, the labels of its edges must (strictly) increase in topological order. Therefore any shortcut $s=(p_i,p_j)\in S$ that shortcuts a distance larger than $1$ can use any time label $t$ of a temporal edge $e = (p_k,p_l,t)\in\tempP$ with $i\leq k\leq l\leq j$. Since $e$ is part of $\tempP$, any temporal edge before $e$ reaches $p_i$ before $t$ and any temporal edge after $e$ departs from $p_j$ after $t$. Therefore, any connection using the temporal subpath $(p_i,...,p_j)\subseteq\tempP$ can now use $(s,t)$. This implies that the temporal edge $(s,t)\in\TDSS$ reduces the same distances as $s$ and adds no new reachabilities, proving the lemma.
\end{proof}

Using \Cref{lm:addingReachableLabels}, we can now use static shortcut sets on the footprint of a temporal path $\tempP$ to achieve a desired diameter $d$. Combining these with the path partitioning explored in \Cref{section:pathIndependence}, we are now able to able to fully utilize construction algorithms for static shortcut sets on temporal graphs $\tempG$ with a directed path as their footprint. Additionally, the computational results regarding the size of a static shortcut set relative to the diameter achieved also hold for this subset of temporal graphs.

%% file: chapters/tdss_expansion.tex
\label{chapter:tdssExpansion}
In this chapter we will present a modified static expansion that allows us to create temporal diameter shortcut sets by translating static shortcut sets into the temporal context. After adding a static shortcut set to our static expansion $\staticTrans(\tempG)$ we deterministically translate the expansion back into a temporal graph with a corresponding temporal shortcut set added to it.

Compared to our other results, this approach can be applied to any directed temporal graph. We will also briefly discuss the size of the resulting temporal shortcut set compared to the static shortcut set on the static expansion. In the following sections we will show the correctness of our modified expansion $\staticTrans(\tempG)$ as defined in \Cref{def:staticExpansion} and the resulting shortcut set.
\begin{restatable}[]{theorem}{shortcutSetsOnExpansion}
    \label{thm:shortcutSetsOnExpansion}
    Let $\tempG=(V,\tempE)$ be a temporal graph and $\staticTrans(\tempG)$ be its static expansion as defined in \Cref{def:staticExpansion}. After adding a static shortcut set $S$ resulting in diameter $k$, we can translate the expansion into a shortcutted temporal graph $\resultTempG = (V,\tempE\cup\tempE^*)$ where $d_\tempG(\resultTempG)\leq k$.
\end{restatable}
We will begin by defining the modified static expansion and observing some properties of shortcut sets on the expansion in \Cref{section:expansionDefinition}, then define how to translate the expansion back into a temporal graph and prove the validity of that graph in \Cref{section:expansionEdgeTranslation} and finally prove the base diameter achieved by the resulting temporal shortcut set and with that \Cref{thm:shortcutSetsOnExpansion} in \Cref{section:expansionDiameterResults}.
\section{Definition}\label{section:expansionDefinition}
We begin by defining the static expansion we use in \Cref{thm:shortcutSetsOnExpansion}. Typically a static expansion creates a timed vertex $v_t$ for every vertex $v$ and time label $t$ in the lifecycle of the graph. Every timed vertex $v_t$ is then connected in temporal order to all other timed vertices based on $v$. A temporal edge $(u,v,t)$ would then correspond to an edge $(u_t,v_{t+1})$ in the static expansion.

Our definition differs slightly from other definitions in literature, which we also described in the introduction. First, we add two vertices per vertex $v$ and time label $t$, for incoming edges and outgoing edges respectively. Secondly, we connect timed vertices that are based on the same vertex $v$ differently. Specifically we only add edges from incoming to outgoing timed vertices.
\begin{restatable}[Static Expansion]{definition}{staticExpansion}
    \label{def:staticExpansion}
    Let $\tempG=(V,\tempE)$ be a temporal graph. For every vertex $v\in V$ we create a \emph{gadget} as follows:
    \begin{enumerate}
        \item For every $t_{min}\leq t\leq t_{max}+1$ create a vertex $v_t^{in}$ and $v_t^{out}$ called \emph{incoming} and \emph{outgoing timed vertices} of $v$. We call $v$ the \emph{base vertex} of a timed vertex $v_t^{in}$ or $v_t^{out}$.
        \item For every incoming timed vertex $v_t^{in}$ add an edge to every outgoing timed vertex $v_{t'}^{out}$, where $t'\geq t$. We call the set of these edges $E_v$ and the union $E_V=\bigcup_{v\in V}E_v$.
    \end{enumerate}
    Given these \emph{gadgets} we create the \emph{static expansion} $\staticTrans(\tempG)=(V_T,E)$ where $V_T$ is the set of all timed vertices based on the gadget construction as explained above and $E=\{(v_t^{out},w_{t'}^{in})\mid (v,w,t)\in\tempE\land t'= t+1\}\cup E_V$ (also see \Cref{fig:staticExpansionExample}).
\end{restatable}
\begin{figure}[h]
    \includegraphics[width=\textwidth]{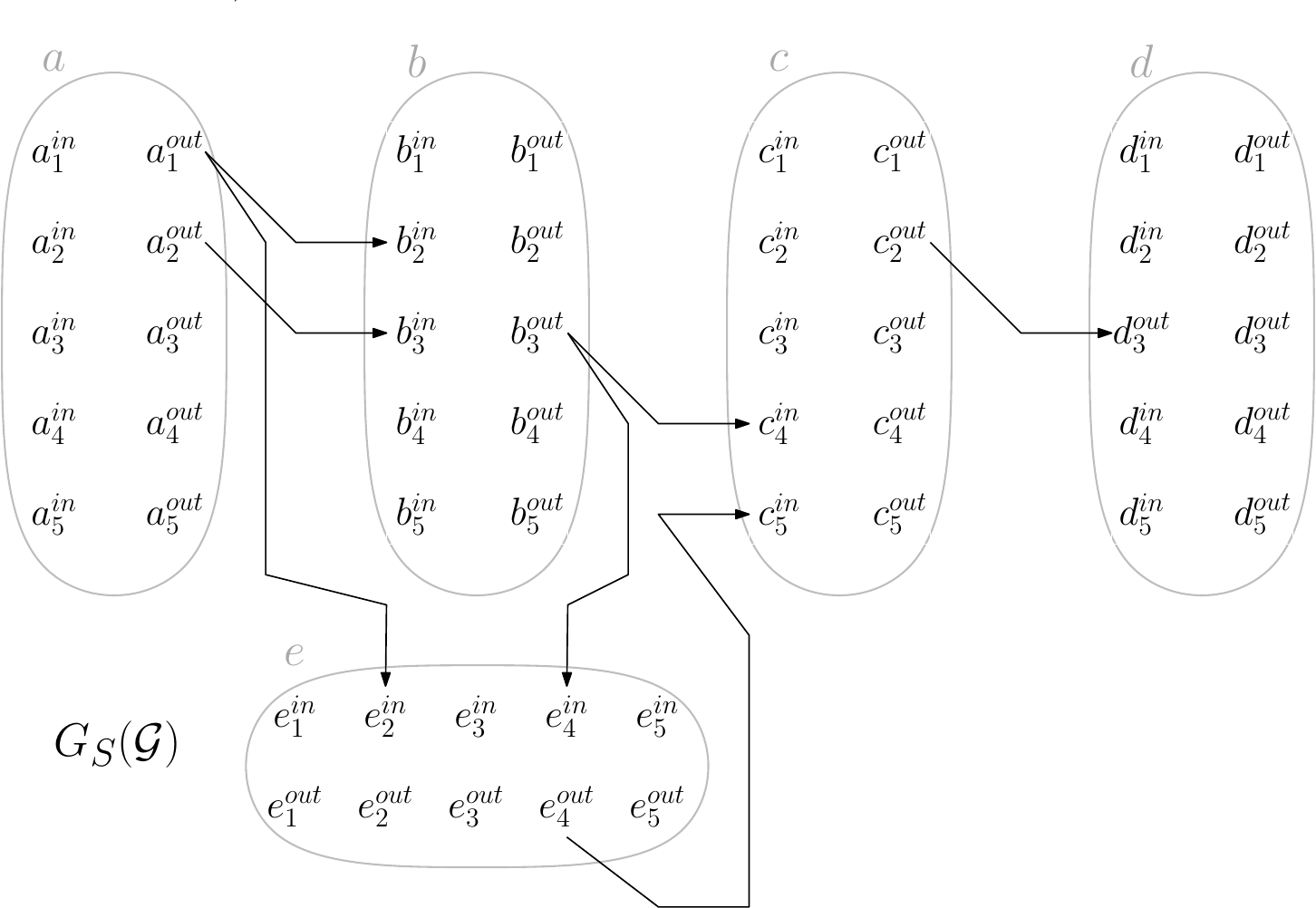}
    \caption{Example of a static expansion $\staticTrans(\tempG)$\footnote{Edges inside of gadgets are omitted from illustrations for better readability. See \Cref{def:staticExpansion} for details.} based on $\tempG$ from \Cref{fig:exampleTemporalGraph}.}
    \label{fig:staticExpansionExample}
\end{figure}
In most other definitions only one timed vertex per time $t$ is used per base vertex. The main benefit of using separate vertices and adding the transitive closure of all gadgets to $\staticTrans(\tempG)$ is that every gadget is traversable in one hop, no matter which vertices are used to traverse it. This also implies a fixed ratio between the diameter of a temporal graph $\tempG$ and its static expansion $\staticTrans(\tempG)$ as stated in \Cref{crl:diameterRatioOriginalToStatic}. Note that the addition of incoming and outgoing timed vertices may not be necessary for this construction to work, although this may be useful for further usage in the future.

\subsection*{Static Shortcut Sets on the Expansion}\label{subsection:shortcutSetsOnExpansion}
The goal of this approach is to be able to use computational results and construction algorithms for static shortcut sets and translate them to the temporal context.

Given a static expansion $\staticTrans(\tempG)$ we can add a static shortcut set $S$ as defined in \Cref{def:staticShortcutSet} to $\staticTrans(\tempG)$. We assume the addition of $S$ results in a diameter of $k$ in $\staticTrans(\tempG)$.

Based on the definition of $S$ we can make a few observations. First of all, recall that a static shortcut set is based on the transitive closure (\Cref{def:staticTransitiveClosure}) of the static graph $G$, i.e., $S\subseteq TC(\tempG)$. We also observe some restrictions for the time labels of the timed vertices used by edges in $S$, namely that a shortcut $s=(\alpha,\beta)\in S$ always travels forward in time.
\begin{restatable}[]{corollary}{expansionShortcuts}
    \label{crl:expansionShortcuts}
    For any edge $s=(\alpha,\beta)\in S$, where $\alpha$ and $\beta$ are any two timed vertices in $\staticTrans(\tempG)$, $t_\alpha<t_\beta$.
\end{restatable}
\Cref{crl:expansionShortcuts} follows directly from the definition of the modified static expansion (\Cref{def:staticExpansion}) and the static shortcut set (\Cref{def:staticShortcutSet}) and the fact that $\alpha$ and $\beta$ cannot lie in the same gadget in $\staticTrans(\tempG)$.

These observations will serve as a basis for proving \Cref{thm:shortcutSetsOnExpansion} throughout the next sections.
\section{Edge Translation}\label{section:expansionEdgeTranslation}
Now that we have expanded our static expansion $\staticTrans(\tempG)$ using a shortcut set $S$, we want to define how each edge in $E\cup S$ can be translated into a temporal edge inside a temporal Graph $\resultTempG=(V, \tempE')$ that uses the same vertex set as $\tempG$ and $\tempE'\supseteq\tempE$.
\newpage
\begin{restatable}[]{definition}{edgeTranslation}
    \label{def:edgeTranslation}
    Given an edge $e = (u,v)\in E\cup S$, we differentiate the following cases:
    \begin{enumerate}[noitemsep]
        \item If $u$ and $v$ are timed vertices of the same gadget, we ignore this edge and do not translate it back, since gadgets translate back to their base vertices in the temporal graph.
        \item If $u$ and $v$ lie in different gadgets, $(u,v)$ translates to $(u,v,t)$, where $t$ is the time label of $u$.
    \end{enumerate}
\end{restatable}
We prove that this results in only valid temporal edges with the following lemma.
\begin{restatable}[]{lemma}{transfoEdgeTranslation}
    \label{lm:transfoEdgeTranslation}
    Given a static expansion $\staticTrans(\tempG)=(V_T, E)$ and a set of shortcut edges $S$ on the expansion, for every edge $e = (u,v)\in S\cup E\setminus E_V$, there exists a valid temporal edge $(u,v,t)$.
\end{restatable}
\begin{proof}[Proof of \Cref{lm:transfoEdgeTranslation}]
    We differentiate between edges in $E$ and edges in $S$.\\
    Firstly, if $e\in E$ then it was part of the original static expansion. Since we exclude $E_V$, which are the edges lying inside of a gadget, by \Cref{def:staticExpansion} this edge must originate from some timed outgoing vertex $v_t^{out}$ and end at some timed incoming vertex $w_{t+1}^{in}$. This would then translate to $(v,w,t)\in\tempE'$, which is a valid temporal edge in $\resultTempG$.
    
    Now let $e = (u,v)\in S$, meaning it was added onto the static expansion as a shortcut edge. Because of \Cref{crl:expansionShortcuts} we know that $e$ does not go back in time. With this information we can translate $e$ into a valid temporal edge $(u,v,t_u)$ similarly to the first case.
\end{proof}
Note that some edges within the shortcutted static expansion $\staticTrans(\tempG)$ may translate to the same temporal edge in $\resultTempG$ (see \Cref{fig:exampleTranslationToSameEdge}).
\begin{figure}
    \centering
    \includegraphics[width=0.7\textwidth]{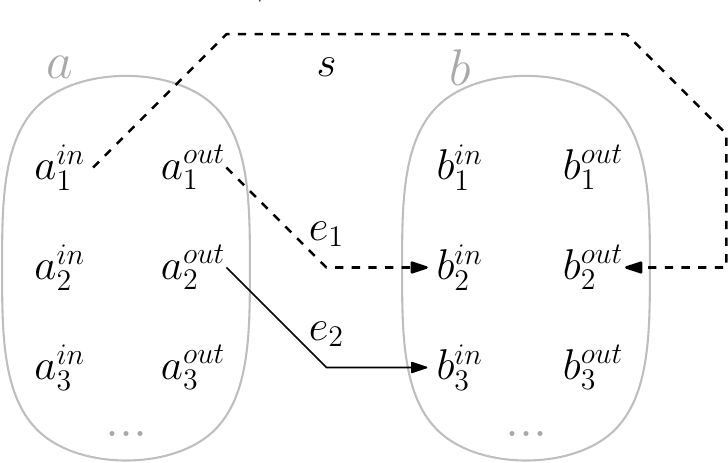}
    \caption{Example of two edges $e_1$ and $s$ translating to the same temporal edge $(a,b,1)$}
    \label{fig:exampleTranslationToSameEdge}
\end{figure}
This means we possibly ignore some translations, if the edge already exists in $\resultTempG$. We will go into further detail regarding the number of shortcuts added to $\resultTempG$ later.
\section{Diameter Results}\label{section:expansionDiameterResults}
Knowing that edges from the static expansion including the shortcut set $S$ translate to a valid temporal graph $\resultTempG$, we now want to determine what results $\resultTempG$ yields regarding its diameter. Hence, we need to analyze the paths in $\staticTrans(\tempG)$ and compare them to the corresponding paths in $\resultTempG$.

Recall that we only regard the base diameter $d_\tempG(\resultTempG)$ of the resulting temporal graph $\resultTempG$, when analyzing the temporal shortcut set $\TDSS$, that was added through the static expansion. For any path $P$, its equivalent path in $\resultTempG$ is formed from the translated edges as defined in \Cref{def:edgeTranslation}. All edges that do not translate (i.e. edges within gadgets) are omitted from the path.
\begin{restatable}[]{lemma}{transfoPathTranslation}
    \label{lm:transfoPathTranslation}
    Given a static expansion $\staticTrans(\tempG)=(V_T,E)$ and a set of shortcut edges $S$ on the expansion, any path $P$ within the resulting graph is longer or equal in length to its equivalent path $P'$ in $\resultTempG$.
\end{restatable}
\begin{proof}[Proof of \Cref{lm:transfoPathTranslation}]
    Firstly, since edges within gadgets start and end within the same base vertex, $P'$ is still a connected path within $\resultTempG$. Because every edge of $P$ translates to exactly one edge for $P'$ and some edges may be omitted, $|P|\geq |P'|$.
    
    Conversely given a path $P'$ in $\resultTempG$, for a corresponding path $P$ in $\staticTrans(\tempG)$ with $S$ added to it $|P'|\leq |P| \leq 2|P'|+1$ holds, since $P$ can use an edge within a single gadget, which would add to its length beyond the length of $P'$ at most $|P'|+1$ times.
\end{proof}
Note that there may exist some path $P$ in $\resultTempG$, which has no corresponding path in $\staticTrans(\tempG)$. Because of this, the actual diameter of $\resultTempG$ may be larger than the base diameter $d_\tempG(\resultTempG)$ in $\resultTempG$. We know that $d_\tempG(\resultTempG)$ is bound by the diameter of $\staticTrans(\tempG)$ after being shortcut, because the reachabilites between gadgets (and between their corresponding base vertices in $\tempG$) are the same and because of the correctness of \Cref{lm:transfoPathTranslation}. We now formally prove \Cref{thm:shortcutSetsOnExpansion} using the information gathered in the previous sections.

\shortcutSetsOnExpansion*

\begin{proof}[Proof of \Cref{thm:shortcutSetsOnExpansion}]
    Let $\staticTrans(\tempG)$ be a static expansion as defined in \Cref{def:staticExpansion} and let $S$ be an added shortcut set resulting in a diameter $k$. We can translate this graph into a valid temporal graph $\resultTempG$ according to \Cref{lm:transfoEdgeTranslation}. According to \Cref{lm:transfoPathTranslation} no path in $\resultTempG$ between vertices reachable in $\tempG$ (and therefore in $\staticTrans(\tempG)$) can be longer than $k$, as $k$ is the diameter of $\staticTrans(\tempG)$ with $S$ added. This implies that $\resultTempG$ has a base diameter $d_\tempG(\resultTempG)$ of at most $k$, proving the theorem.
\end{proof}

\section{Discussion \& Further Observations}\label{section:expansionDiscussion}
Given \Cref{thm:shortcutSetsOnExpansion}, we have proven the basic construction and correctness of our approach. In this section we present some further observations and analyze the size of the temporal shortcut set in the resulting temporal graph.

First, the diameter ratio between the resulting temporal graph $\resultTempG$ and the static expansion $\staticTrans(\tempG)$ shortcutted with $S$ gives us a few more insights. We can observe a fixed ratio between the diameter of the original temporal graph $\tempG$ and its static expansion $\staticTrans(\tempG)$.
\begin{restatable}[]{corollary}{diameterRatioOriginalToStatic}
    \label{crl:diameterRatioOriginalToStatic}
    Given a temporal graph $\tempG$ with diameter $k$ the static expansion $\staticTrans(\tempG)$ as defined in \Cref{def:staticExpansion} has a diameter of $2k+1$ 
\end{restatable}
\Cref{crl:diameterRatioOriginalToStatic} follows directly from \Cref{def:staticExpansion}. More specifically, since any gadget in $\staticTrans(\tempG)$ can be traversed in one hop, any path in $\tempG$ is extended by one edge per base vertex or gadget in $\staticTrans(\tempG)$ on that path.

Given that \Cref{thm:shortcutSetsOnExpansion} only provides an upper bound for the base diameter of $\resultTempG$, we explored wether one could achieve a base diameter $k$ in $\resultTempG$ with a shortcut set $S$ achieving a diameter larger than $k$ in $\staticTrans(\tempG)$. 

Unfortunately there exist some shortcut sets $S$ on the expansion achieving a diameter $k$ with $\resultTempG$ also having a base diameter $k$. In \Cref{fig:diameterRatioCounterExample} $S_1$ and $S_2$ both achieve a diameter of $k=5$. After translating both constructions back into a temporal graph, the graph translated from $S_1$ has a base diameter of $k = 5$, while the other temporal graph has a base diameter of $k=2$. Given that the temporal graph, on which the static expansions in \Cref{fig:diameterRatioCounterExample} are based, has diameter $4$, the diameter ratio given in \Cref{crl:diameterRatioOriginalToStatic} holds for the static expansion shortcutted with $S_2$. The existence of $S_1$ implies that the diameter ratio does not hold for any static shortcut set though. Hence, the ratio given in \Cref{crl:diameterRatioOriginalToStatic} can not be applied between the shortcutted static expansion and the resulting temporal graph $\resultTempG$.
\begin{figure}
    \centering
    \includegraphics[width=\textwidth]{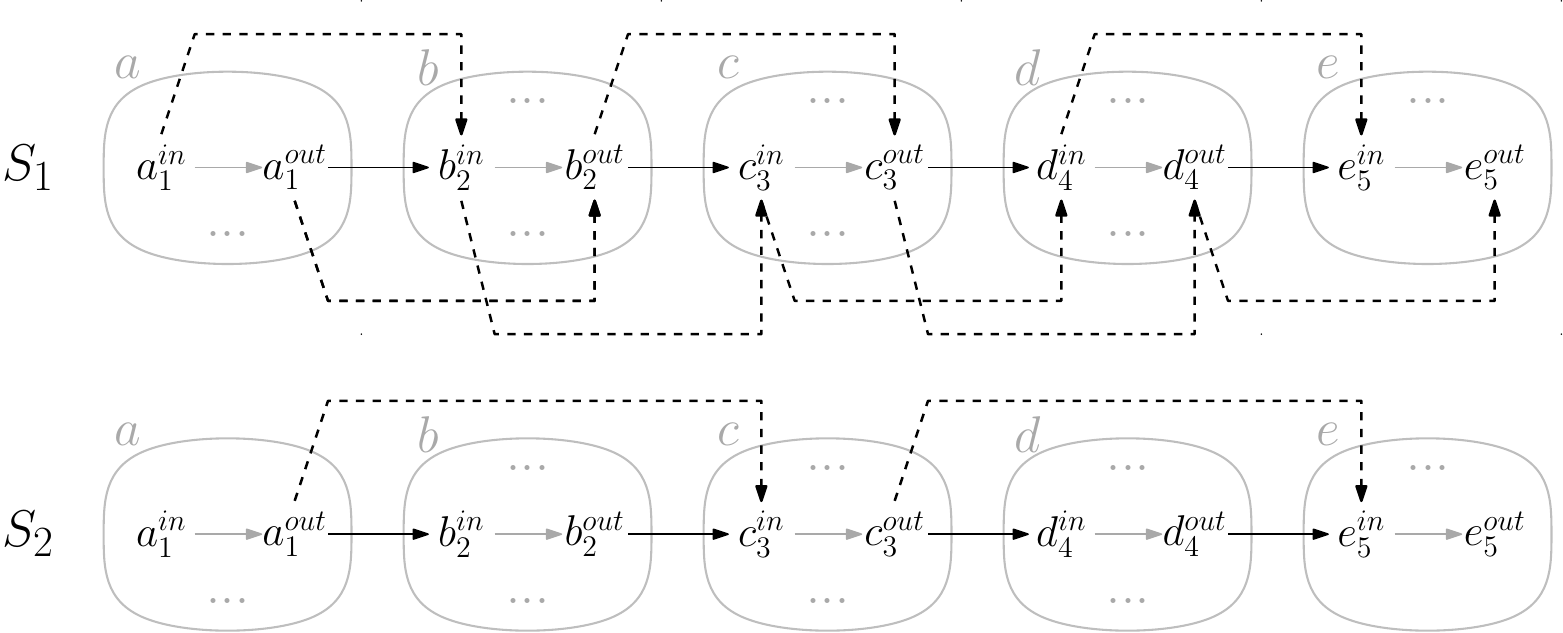}
    \caption{The shortcut sets $S_1$ and $S_2$ on a static expansion $\staticTrans(\tempG)$ both achieving a diameter of $k=5$.\footnote{Some timed vertices are omitted from the gadgets to improve readability.}}
    \label{fig:diameterRatioCounterExample}
\end{figure}

On the other hand the existence of $S_1$ (see \Cref{fig:diameterRatioCounterExample}) implies that there exist some shortcut sets, specifically those where many shortcuts or even every shortcut connects an outgoing to an incoming timed vertex, that achieve a much better base diameter in $\resultTempG$. As stated in \Cref{crl:diameterRatioOriginalToStatic} the best possible diameter ratio between $\staticTrans(\tempG)$ and the resulting temporal graph $\resultTempG$ would be $2k+1$ to $k$. To achieve such a shortcut set at all times, we would need to adjust the construction of $S$. Since we want to be able to use existing static shortcut set constructions and algorithms, this is not an option in our case. Note that a shortcut set $S$ that achieves diameter $k$ in the expansion may still result in a base diameter smaller than $k$ in the resulting temporal graph $\resultTempG$ using our approach.

\subsection*{Size of the resulting TDSS}\label{subsection:expansionSizeResults}
We analyze how many edges are added to our final temporal graph $\resultTempG$ through the additition of $S$ and the translation as defined in \Cref{def:edgeTranslation}. Additionally, we translate the computational results from static shortcut sets into the temporal context using our approach. For edges added through $S$ there exist three cases:
\begin{enumerate}[noitemsep]
    \item The edge is translated into a new temporal edge, i.e., that was not part of the original graph $\tempG$,
    \item 1. is the case, but there exists another edge that translates to the same edge or
    \item the edge translates to an original edge of $\tempG$.
\end{enumerate}

For the size of the final shortcut set in $\resultTempG$ we only count the edges from case 1, since edges from case 2 are omitted in $\resultTempG$ and edges from case 3 already existed in the original temporal graph $\tempG$. We can obviously observe that the resulting temporal shortcut set is no larger than $S$.
\begin{restatable}[]{corollary}{translatedShortcutSetSize}
    \label{crl:translatedShortcutSetSize}
    Let $\tempG$ be a temporal graph and $\staticTrans(\tempG)$ its static expansion. After adding a static shortcut set $S$ to $\staticTrans(\tempG)$ and translating $\staticTrans(\tempG)$ into a temporal graph $\resultTempG$ (as defined in \Cref{def:edgeTranslation}) the number of added shortcuts in $\resultTempG$ compared to $\tempG$ is at most $|S|$.
\end{restatable}
This implies that we can at worst use the computational results for static shortcut sets based on the size of $\staticTrans(\tempG)$ to compute results for temporal shortcut sets in $\tempG$. We can define the following bounds on the size of the temporal shortcut set in $\resultTempG$ based on the number of vertices $n\cdot t_{max}$\footnote{$n$ is the number of vertices in $\tempG$.} in $\staticTrans(\tempG)$.
\begin{restatable}[]{lemma}{temporalShortcutSetSize}
    \label{lm:temporalShortcutSetSize}
    Let $\tempG=(V,\tempE)$ be a temporal graph and let $t_{max}$ be the largest time label in $\tempE$. By adding $\bigOlog{n\cdot t_{max}}$ shortcuts to $\tempG$ we achieve a base diameter $d_\tempG(\resultTempG)$ of at most $\bigO{(n\cdot t_{max})^{1/3}}$ and at best $\bigOmegalog{(n\cdot t_{max})^{1/4}}$ in the resulting graph $\resultTempG$.
\end{restatable}
\begin{proof}[Proof of \Cref{lm:temporalShortcutSetSize}]
    Let $\tempG=(V,\tempE)$ be a temporal graph with $|V|=n$ and $\staticTrans(\tempG)$ its static expansion. By \Cref{def:staticExpansion}, we know that $\staticTrans(\tempG)$ has $n\cdot t_{max}$ vertices. Any static graph, including $\staticTrans(\tempG)$, with $x$ vertices can be shortcut with $\bigOlog{x}$ vertices to a diameter of at most $\bigO{x^{1/3}}$ and at best $\bigOmegalog{x^{1/4}}$ (see \Cref{section:staticRelatedWork}). Combined with \Cref{crl:translatedShortcutSetSize} and the size of $\staticTrans(\tempG)$ we achieve a base diameter $d_\tempG(\resultTempG)$ of at most $\bigO{(n\cdot t_{max})^{1/3}}$ and at best $\bigOmegalog{(n\cdot t_{max})^{1/4}}$ in the resulting graph $\resultTempG$ after adding $\bigOlog{n\cdot t_{max}}$ shortcuts to $\staticTrans(\tempG)$ and translating it back.
\end{proof}

\Cref{lm:temporalShortcutSetSize} applies the ratio found in $\Cref{crl:translatedShortcutSetSize}$ using current state of the art computational results for static shortcut sets. Using this ratio, improved results for static shortcut sets can also be applied to temporal shortcut sets in a similar way in the future.

%% file: conclusions/conclusions.tex
\label{chapter:conclusions}
In this thesis we explored the concept of diameter shortcut sets in temporal graphs. Due to the nature of temporal graphs, the problem of temporal diameter shortcut sets is much more complex compared to its counterpart in static graphs. Specifically the additional restrictions for temporal journeys and the lack of transitivity in temporal reachability complicate the construction of temporal shortcut sets.
\begin{enumerate}
    \item We first explored how to define a Temporal Diameter Shortcut Set (TDSS) by looking at similar concepts in static graphs. We concluded that a reachability constraint on a temporal shortcut set is not feasible and defined the concept for temporal graphs accordingly.
    \item We analyzed the construction of temporal shortcut sets on temporal paths and temporal graphs with directed paths as their footprint. We proved that the usage of construction algorithms for static shortcut sets leads to an optimal construction for temporal shortcut sets as well.
    \item Finally, we presented a translation approach using a modified static expansion. We were able to translate static shortcut set constructions on the static expansion into valid results in the temporal context. This approach may be used for any directed temporal graph.
\end{enumerate}

Our results laid important groundwork for the further exploration of this field in temporal graphs. We made many fundamental observations which often restrict the possible approaches in comparison to solutions on regular static graphs. In the following we outline some open questions that were introduced during the writing of this thesis.
\section*{Open Questions \& Future Work}\label{section:futureWork}
We introduced a definition of temporal diameter shortcut sets that differs greatly from the existing static shortcut set. The decision to ignore newly added reachabilities was vital to enable the results we later presented. It may be possible though to construct temporal shortcut sets that are not less optimal than the ones we construct, without this constraint.

It may also be possible to expand the independence of temporal paths on underlying static paths, presented in \Cref{chapter:tdssPaths}, to more complex graphs like trees or DAGs. We proved that the reachabilities of the constructed static shortcut set are not expanded after adding time labels to the shortcuts, converting it into a temporal shortcut set. Of course the temporal shortcut sets created this way, enforce a certain diameter not just regarding the original reachabilities of the shortcutted temporal graph. If it were possible to expand this approach to more complex graphs, the restriction on our definition of temporal diameter shortcut sets may also not be necessary.

There are also some open questions regarding our translation approach using the modified static expansion. As mentioned earlier the addition of incoming and outgoing timed vertices may not be necessary for the construction to work. Reducing the construction to using only a single timed vertex per time $t\in\N$ would decrease the complexity of the static expansion graph and should therefore be explored further. 

We also only briefly explored the size of the resulting temporal shortcut set relative to the static shortcut set added to the expansion. It would be beneficial to explore the size ratio between the two further. This would enable us to apply future computational results for static shortcut sets to temporal shortcut sets more accurately. Of course additional research regarding the amount of temporal shortcuts needed to achieve a certain diameter independent of a given construction approach is needed.

Our translation approach inherently also restricts the possible shortcuts added. Due to the construction of the static expansion and the restriction of reachability in static shortcut sets no shortcuts with time labels outside of the lifecycle of the original temporal graph $\tempG$ can be added. The reachability constraint in the static expansion also does not allow us to add shortcuts with time labels later than the latest departure time. Through modifying the static expansion or the translation rules we presented, it may be possible to remove some of these restrictions, which may lead to improved results for the temporal shortcut set.

%% file: core/declaration_of_authorship/declaration_of_authorship.tex
I hereby declare that this thesis is my own unaided work. All direct or indirect sources used are acknowledged as references.\\[6 ex]

\begin{flushleft}
    Potsdam, \today
    \hspace*{2 em}
    \raisebox{-0.9\baselineskip}
    {
        \begin{tabular}{p{5 cm}}
            \hline
            \centering\footnotesize\printAuthor
        \end{tabular}
    }
\end{flushleft}